\documentclass[aos]{imsart}

\RequirePackage{amsthm,amsmath,amsfonts,amssymb}
\RequirePackage[numbers]{natbib}
\usepackage{graphicx}
\usepackage{algorithm2e,mathabx}
\usepackage{xcolor}
\RequirePackage[colorlinks,citecolor=blue,urlcolor=blue]{hyperref}%% uncomment this for coloring bibliography citations and linked URLs

\startlocaldefs
\definecolor{asparagus}{rgb}{0.53, 0.66, 0.42}

\newtheorem{thm}{Theorem}[section]
\newtheorem{lem}[thm]{Lemma}
\newtheorem{prop}[thm]{Proposition}
\newtheorem{cor}[thm]{Corollary}

\theoremstyle{definition}
\newtheorem{defn}[thm]{Definition}

\newtheorem{ex}[thm]{Example}

\theoremstyle{remark}
\newtheorem{rem}[thm]{Remark}

\DeclareMathOperator{\tr}{tr} 
\DeclareMathOperator{\cov}{\mathbb{V}}

\DeclareMathOperator{\argmin}{argmin} 
 
\DeclareMathOperator{\K}{\mathbf{K}}

\newcommand{\E}{\mathbb{E}}
\newcommand{\cE}{\mathcal{E}}

\newcommand{\cS}{\mathcal{S}}

\newcommand{\F}{\mathbb{F}}

\renewcommand{\P}{\mathbb{P}}

\newcommand{\R}{\mathbb{R}}

\newcommand{\cX}{\mathcal{X}}

\newcommand{\bs}{\boldsymbol}
\newcommand{\mtp}{{\text{MTP}}_{2}}

\newcommand{\cd}{\,|\,}
\newcommand{\<}{\langle}
\renewcommand{\>}{\rangle}
\renewcommand{\-}{\setminus}
\newcommand{\ep}{{\rm A}} 
\newcommand{\mpl}{{\rm M}_+}

\newcommand{\bic}{\textsc{BIC}}

\endlocaldefs

\begin{document}

\begin{frontmatter}
\title{Locally associated graphical models \\ and mixed convex exponential families}
%\title{A sample article title with some additional note\thanksref{T1}}
\runtitle{Locally associated graphical models}
%\thankstext{T1}{A sample of additional note to the title.}

\begin{aug}
%%%%%%%%%%%%%%%%%%%%%%%%%%%%%%%%%%%%%%%%%%%%%%%
%% Only one address is permitted per author. %%
%% Only division, organization and e-mail is %%
%% included in the address.                  %%
%% Additional information can be included in %%
\author[A]{\fnms{Steffen} \snm{Lauritzen}\ead[label=e1]{lauritzen@math.ku.dk}}
\and
\author[B]{\fnms{Piotr} \snm{Zwiernik}\ead[label=e2]{piotr.zwiernik@utoronto.ca}}
%\and
%\author[B]{\fnms{???} \snm{???}\ead[label=e3,mark]{???@???}}
%%%%%%%%%%%%%%%%%%%%%%%%%%%%%%%%%%%%%%%%%%%%%%
%% Addresses                                %%
%%%%%%%%%%%%%%%%%%%%%%%%%%%%%%%%%%%%%%%%%%%%%%
\address[A]{University of Copenhagen, \printead{e1}}

\address[B]{University of Toronto, \printead{e2}}
\end{aug}

\begin{abstract}
The notion of multivariate total positivity has proved to be useful in finance and psychology but may be too restrictive in other applications. In this paper we propose a concept of local association, where highly connected components in a graphical model are positively associated and study its properties. Our main motivation comes from gene expression data, where graphical models have become a popular exploratory tool. The models are instances of what we term \emph{mixed convex exponential families} and we show that a \emph{mixed dual likelihood estimator} has simple exact properties for such families as well as asymptotic properties similar to the maximum likelihood estimator. We further relax the positivity assumption by penalizing negative partial correlations in what we term the \emph{positive graphical lasso}. Finally, we develop a GOLAZO algorithm based on  block-coordinate descent that applies to a number of optimization procedures that arise in the context of graphical models, including the estimation problems described above. We derive results on existence of the optimum for such problems. 
\end{abstract}



\end{frontmatter}

\section{Introduction and summary}

It has been illustrated recently in a number of publications that explicitly incorporating positive dependence constraints can be useful for modelling in various contexts where components are naturally positively associated (e.g.\ finance or psychology) \cite{LUZ,agrawal2019covariance,MTP2Markov2015,lauritzen2019total,slawski2015estimation}. The main distinctive feature of this line of work as opposed to more classical literature on positive dependence is that they link to techniques used in high-dimensional statistics and graphical models using the positivity constraint as an implicit regularizer. 

In the Gaussian case, a natural positivity constraint is that all partial correlations are non-negative or, equivalently, the inverse covariance matrix is an M-matrix, $(\Sigma^{-1})_{ij}\leq 0$ for all $i\neq j$. Optimizing a loss function under this restriction typically results in a sparse estimate, which was  the driving idea in \cite{slawski2015estimation}. For standard stock market datasets this may lead to an estimate that gives both a sparser graph and a higher value of the likelihood function than estimates from the graphical lasso approach \cite{rossell2020dependence}. 

Although useful, this global positivity constraint is often too restrictive. In this paper we propose and study natural relaxations of the condition. With an underlying graph representing the dependence structure between the variables, we will require that highly connected components are positively dependent, in the precise sense that variables in the same clique are \emph{associated} \cite{esary1967association}, which in the Gaussian case is equivalent to having a covariance matrix with positive entries \cite{pitt1982positively}. Unfortunately, maximum likelihood estimation (MLE) in this type of model is problematic as the likelihood function may get complicated. The model is an instance of what we term a \emph{mixed convex exponential family}. We develop an associated \emph{mixed dual estimator} (MDE) which overcomes the problems faced by the MLE. The MDE can be found by solving two convex optimization problems, and has asymptotic properties similar to the MLE. We note that Maliutov \emph{et al.} \cite{malioutov2006walk} suggest a different relaxation of the $\mtp$ condition, unrelated to ours. 

In this situation, when the underlying graph is not known, we also consider a further relaxation of local association through what we name the \emph{positive graphical lasso} which penalizes large negative partial correlations. As for the now classical graphical lasso, this will typically identify a sparse structure. % in the concentration matrix. 

\subsection{A motivating example}
A motivating problem is the exploratory analysis of gene expression data. The graphical lasso has become a standard technique for estimating gene expression networks. While constructing and interpreting the network, researchers often focus on positive co-expression (e.g. \cite{mason2009signed}), where pairs of genes show a proportional expression pattern across samples. Also, in various scenarios it has been observed that positively co-expressed genes within the same pathway tend to cluster close together in the pathway structure, while negatively correlated genes typically occupy more distant positions; see e.g.\  \cite{kharchenko2005expression,wei2006transcriptional}. 
 We shall later, in Section~\ref{sec:gene} analyse publicly available microarray expression data profiling umbilical cord tissue in a study of fetal inflammatory response (FIR); cf.\ \cite{cohen2007perturbation,costa2016umbilical}. 
From an initial set of 12,093 genes with reliable expression, Costa and Castelo \cite{costa2016umbilical} identified 1,097 as differentially expressed between FIR-affected and unaffected infants, from which 592 were upregulated in FIR. This subset of 592 upregulated genes was significantly enriched by 136 genes involved in the innate immune response (\cite{breuer2013innatedb}) and we shall focus our analysis on this subset of 136 genes.

A typical approach to explicitly model positive co-expression is by building weighted gene co-expression  networks where correlations are mapped monotonely from $[-1,1]$ to $[0,1]$ and then thresholded. This approach is subject to standard problems with  building co-expression networks based on correlations alone not taking the effect of other genes into account. Our approach is based on partial correlations and so may provide more meaningful estimates of the underlying network.

\subsection{An optimization algorithm}
An important first step in modelling large systems that satisfy some positive dependence constraint is to reduce to a sparser representation without loosing the positive dependence information. In this respect, $\ell_1$-regularized approaches do not work well, since they treat positive and negative partial correlations in an equal manner. In this paper we propose a version where only negative partial correlations are penalized. 

This approach is developed further in a general GOLAZO\footnote{Pronounced \emph{goh-lah-soh}, like the  Spanish word \emph{golazo} but without the Castilian lisp.} algorithm (Graphical Oriented LAZy Optimization) 
where a penalty of the form $$\sum_{i\neq j} \max\{L_{ij}K_{ij},U_{ij}K_{ij}\}$$ is employed to obtain sparse estimation of $K=\Sigma^{-1}$. Here the penalty parameters $L,U$ satisfy $-\infty\leq L_{ij}\leq 0\leq U_{ij}\leq +\infty$; including zero and infinite values solves several optimization problems proposed in this paper and a number of related problems in graphical models. The advantage of our general approach is that it provides a detailed analysis of convergence and existence of the optimum. 

\subsection{Main contributions and structure of paper}
The main contributions in this paper are
\begin{enumerate}
    \item [(i)] A class of Gaussian graphical models with relaxation of positivity restrictions, either via a graph or via a positive lasso-type penalty, or both;
    \item [(ii)] A general framework for mixed convex exponential families with an associated method of estimation that has asymptotic properties similar to maximum likelihood but more favourable computational properties;
    \item [(iii)] The GOLAZO algorithm for solving a range of estimation problems associated with graphical models, including positivity restrictions. The associated R package \texttt{golazo} is available on GitHub.
\end{enumerate}
The structure of the paper is as follows. In Section~\ref{sec:loc_ass} we study the basic properties of locally associated distributions and discuss their relation to other models involving positivity. 
Section~\ref{sec:mixed} introduces and studies mixed convex exponential families. The mixed dual estimator is introduced in Section~\ref{sec:estimation} and applied to locally associated Gaussian graphical models (laGGMs) in Section~\ref{sec:EP}. 
The asymptotic properties of the estimator are established in Section~\ref{sec:asym}.
In Section~\ref{sec:pglasso} we focus on learning the structure of a laGGM and introduce the positive graphical lasso for this purpose. We derive the GOLAZO algorithm in Section~\ref{app:golazo} and argue that a number of optimization problems in graphical models can be seen as special instances and hence solved by this algorithm. 
Section~\ref{sec:appls} concludes the paper by applying the methodology to two  datasets.

\section{Locally associated distributions}\label{sec:loc_ass}

In this section we define local association and locally associated Gaussian graphical models linking to other relevant statistical models.

\subsection{Definition and basic properties}

We say that a function $f:\R^d\to \R$ is increasing if $x\leq x'$ (coordinatewise) implies that $f(x)\leq f(x')$. A $d$-dimensional random vector $X$ is (positively) \emph{associated} if for any two increasing functions $f,g:\R^d\to \R$, the  covariance $\cov(f(X),g(X))$ is non-negative; for basic properties of this notion see \cite{esary1967association}. 

In general, association is hard to check and \cite{KarlinRinott80} introduced the strictly stronger notion of multivariate total positivity  which in the Gaussian case is equivalent to the covariance $\Sigma$ being an inverse M-matrix \cite{karlinGaussian,LUZ}. 
We recall that $K=\Sigma^{-1}$ is an M-matrix if it is positive definite and $K_{ij}\leq 0$ for all $i\neq j$; so this condition corresponds to assuming that all partial correlations are non-negative. 

In the Gaussian case there is a simple condition for association, as stated in the theorem below:
\begin{thm}[Pitt \cite{pitt1982positively}]\label{thm:gauss_ass}
		Suppose $X$ is a Gaussian vector with covariance matrix $\Sigma$. It then holds that
		$
		X$ is associated if and only if $\Sigma_{ij}\geq 0$ for all $i,j$.
\end{thm}
%\noindent Here and later the matrix inequality is to be interpreted coordinatewise. 
Positive semi-definite matrices with non-negative  entries are also called doubly non-negative matrices.
If $\Sigma$ is an inverse M-matrix, then it is doubly non-negative, but the reverse implication does not hold. 

The requirement that a distribution is associated is  strong and we wish to relax this in relation to a local structure given by a graph. We build  on the standard terminology for undirected and directed graphs as given, for example, in \cite{lau96}. If $G=(V,E)$ is an undirected  graph with vertex set $V=\{1,\ldots,d\}$ and edge set $E$, a \emph{complete subset} of $G$ is any subset $C$ of $V$ such that any two vertices  $i,j\in C$ are connected by an edge, that is, $ij\in E$. A  \emph{clique} is a complete subset  that is maximal with respect to inclusion. Let $X=(X_1,\ldots,X_d)$ be a random vector and fix a graph $G=(V,E)$. For any $C\subseteq V$, by $X_C$ denote the subvector of $X$ with entries $X_i$, $i\in C$. 
\begin{defn}\label{def:locass}
 The random vector $X$ is said to be \emph{locally associated} w.r.t.\ $G$ if  it holds for any clique $C$ of $G$ that the subvector $X_C$ is associated. 
\end{defn}

Denote by $\ep(G)$ the set of covariance matrices of Gaussian vectors that are locally associated with respect to $G$. By Theorem~\ref{thm:gauss_ass} these are the positive definite matrices $\Sigma$ such that $\Sigma_{ij}\geq 0$ for all edges $ij\in E$.

\subsection{Locally associated Gaussian graphical models}\label{sec:laggm}

Our main interest lies in locally associated distributions for Gaussian graphical models. We say that a distribution of a random vector $X$ is \emph{Markov} with respect to $G$, or $M(G)$, if $X$ satisfies global Markov properties over the graph $G$; for more on graphical models see \cite{lau96}. A Gaussian vector $X$ with covariance matrix $\Sigma$ is Markov with respect  to $G$ if $(\Sigma^{-1})_{ij}=0$ for all $ij\notin E(G)$. The distributions that are Markov and locally associated with respect to $G$ are denoted by $\mpl(G)$. We refer to $\mpl(G)$ as a \emph{locally associated Gaussian graphical model (locally associated GGM)}. By definition we have 
$$\mpl(G)\;\;=\;\;\ep(G)\cap M(G).$$

 The set of locally associated Gaussian distributions that are Markov with respect to a graph is given as the intersection of a set that is convex in $K$ with  a set that is convex in $\Sigma$. The intersection is typically neither convex in $K$ nor in $\Sigma$ but a locally associated GGM is an instance of what we shall term a \emph{mixed convex exponential family}; see Section~\ref{sec:mixed} below.

\subsection{Positive linear systems and factor models}

In this section we link locally associated Gaussian graphical models to a broad class of models that includes, for example, factor analysis models with non-negative loadings. 

Recall that a Gaussian model over a directed acyclic graph (DAG) $D$ has the linear structural representation
$$
Y_i\;=\;\sum_{i\to j} \lambda_{ij} Y_j+\epsilon_i\qquad\mbox{for all }i\in V,
$$
where $(\epsilon_1,\ldots,\epsilon_d)$ is a mean-zero vector of Gaussian independent noise terms and $\lambda_{ij}\in \R$. Write $M(D)$ for the class of all such distributions parameterized by $\Lambda=[\lambda_{ij}]$ and the variances of $\epsilon_i$. Moreover, $M_+(D)$ denotes the subset of $M(D)$ where $\Lambda\geq 0$, i.e.\ where all regression coefficients are non-negative. 
\begin{prop}\label{prop:DAGass}
    Suppose that the distribution of a zero-mean Gaussian $Y$ lies in $M_+(D)$. Then  $Y$ is associated and so is each margin of $Y$. 
\end{prop}
\begin{proof}
Since $(I-\Lambda)Y=\epsilon$ then $\cov(Y)=(I-\Lambda)^{-1}\Omega(I-\Lambda)^{-T}$, where $\Omega$ is a diagonal matrix with the variances of $\epsilon$ on the diagonal and $(I-\Lambda)$ being invertible by acyclicity of $D$. Since $(I-\Lambda)^{-1}=I+\Lambda+\Lambda^2+\ldots$ and $\Lambda\geq 0$ we conclude that $\cov(Y)$ has only non-negative entries. This % of course applies also to each submatrix, which 
concludes the proof. 
\end{proof}
Proposition~\ref{prop:DAGass} shows that associated distributions contain the interesting family of positive DAG models $M_+(D)$, potentially with some nodes unobserved. Factor analysis models with non-negative loadings form a particular instance of margins of DAG models of the form $M_+(D)$. Recall that the  factor analysis model $\F_{d,s}$  is the family  of multivariate Gaussian distributions with an arbitrary mean and whose covariance matrix $\Sigma$ is of the form $\Sigma=\Delta+\Lambda\Lambda^T$ with a positive diagonal matrix $\Delta$ and $\Lambda\in \R^{d\times s}$. We write  $\F_{d,s}^+$ if the loading matrix $\Lambda$ is restricted to have non-negative entries.

One of the  standard arguments for why $\mtp$ distributions may be useful in statistical modelling is that they  contain the one factor model with non-negative loadings, $\F_{d,1}^+$. This observation and the corresponding link to the  Capital Asset Pricing Model was used in \cite{agrawal2019covariance} to argue why $\mtp$ distributions are particularly suitable for modelling financial data. However, it is easy to show by explicit examples that distributions in $\F_{d,s}^+$ for $s>1$ need not be $\mtp$. In this context, the fact that all of them are (globally) associated may provide useful regularization procedures in applications where factor analysis models with non-negative loadings are expected to perform well; see Section~\ref{sec:appls} for some evidence.

\subsection{Gaussian distributions and Gaussian copulas}

Since association is preserved after applying a strictly increasing function $\phi_i:\R\to \R$ to each $X_i$, our definition of local association naturally extends to Gaussian copula models, as in Proposition~\ref{prop:cop} below. A $d$-dimensional random vector $Y$ has a non-paranormal distribution if there exist strictly  increasing functions $\phi_i:\R\to \R$ for $i=1,\ldots,d$ such  that the vector $\phi(Y):=(\phi_1(Y_1),\ldots,\phi_d(Y_d))$ has a Gaussian distribution. 
\begin{prop}\label{prop:cop} If $G=(V,E)$ is an undirected graph and $Y$ has a non-paranormal distribution then $Y$ is locally associated with respect to $G$ if and only if $\phi(Y)$ is in $\ep(G)$. Moreover,  $Y$ is Markov with respect to $G$ if and only if $\phi(Y)$ is.
\end{prop}

\section{Mixed convex exponential families}\label{sec:mixed}

It is useful to see locally associated Gaussian graphical models as a special case of a more general type of models. Consider a random variable $X$ with values in a general state-space $\cX$. Suppose that the distribution of $X$ is in a minimally represented regular exponential family $\cE=\{P_{\bs\theta}\cd\bs\theta\in \Theta\}$ with canonical statistic $\bs t:\cX\mapsto \R^k$ and canonical parameter $\bs \theta$. This means that the density function $p(\bs x;\bs \theta)$ of the distribution $P_{\bs \theta}$ with respect to some underlying measure $\nu$ on $\cX$ takes the form
\begin{equation}\label{eq:expfam}
p(\bs x;\bs \theta)\;=\;\exp\{\<\bs \theta,\bs t(\bs x)\>-A(\bs \theta)\}\qquad \mbox{for }\bs\theta\in \Theta,
\end{equation}
 where $\nu\{x:\<\lambda,t(x)\>=c\}=0$ if $\lambda\neq 0$. 
The space of canonical parameters
$$\Theta\;:=\;\left\{\bs\theta\in \R^k: \int_\cX \exp\big\{\<\bs\theta,\bs t(\bs x)\>\big\}\,\nu({\rm d} \bs x)<\infty\right\}$$ is an open set in $\R^k$ and the cumulant function $A:\Theta\to \R$ is strictly convex and smooth. The map $\mu$ between canonical parameter $\bs \theta\in \Theta$ and the mean parameter $\bs \mu\in M$ satisfies $$\mu(\bs \theta)=\nabla A(\bs \theta)$$ and  establishes a bijection between $\Theta$ and $M$. Moreover, $M$ is the interior of the convex hull of $\bs t(\cX)$; see any of the references \cite{barndorff:78,brown:86,sundberg} for more details. The inverse map is denoted by $\theta$, that is, $\bs \theta=\theta(\bs \mu)$.

Suppose we split the sufficient statistics into two subvectors $\bs t(\bs x)=(\bs u,\bs v)$ of dimension $r,s$ where $r+s=k$. Let $\bs \theta=(\bs \theta_u,\bs \theta_v)$, $\bs \mu=(\bs \mu_u,\bs \mu_v)$ be the corresponding splits in the canonical and in the mean parameter. In analogy with $\bs\mu =\mu(\bs\theta)$, we also use the notation $\mu_u(\bs\theta)=\bs \mu_u$, $\mu_v(\bs\theta)=\bs \mu_v$, $\theta_u(\bs\mu)=\bs\theta_u$, and $\theta_v(\bs\mu)=\bs\theta_v$. For example, $\mu_u(\bs\theta)$ is the composition of $\mu:\Theta\to M$ with the projection $(\bs\mu_u,\bs\mu_v)\mapsto \bs \mu_u$.

By  \cite[Theorem~8.4]{barndorff:78}, the pair $(\bs \mu_u,\bs \theta_v)$ forms an alternative  parametrization for the exponential family $\cE$ called the \emph{mixed parametrization}. The parameters $\bs \mu_u$ and $\bs \theta_v$ are \emph{variation independent}, that is, the parameter space for $(\bs\mu_u,\bs\theta_v)$ is the Cartesian product space $M_{u}\times \Theta_v$ where $\Theta_v$ is the projection of $\Theta$ on $\bs \theta_v$ and $M_u$ is the projection of the space of mean parameters $M$ onto $\bs\mu_u$. So we may without ambiguity write
\[\cE=\{P_{\bs\theta}\cd\bs\theta\in \Theta\}=\{P_{\bs\mu}\cd\bs\mu\in M\}=\{P_{(\bs\mu_u,\bs\theta_v)}\cd\bs\mu_u\in M_u,\bs\theta_v\in \Theta_v\}.\]
We may thus consider the model $\cE'$ obtained from $\cE$ by a convex restriction on $\bs \mu_u$ and a convex restriction on $\bs \theta_v$. More precisely, we define:
%\color{purple}
\begin{defn}\label{def:mixed_convex}
{Fix a mixed parametrization $(\bs \mu_u,\bs \theta_v)\in M_u\times \Theta_v$ of the exponential family $\cE$. The model $\cE'$ is called a \emph{mixed convex} submodel of $\cE$ and a \emph{mixed convex exponential family} if it consists of all mixed parameters $(\bs \mu_u,\bs \theta_v)\in M_u'\times \Theta_v'$, where $M_u'$ and $\Theta_v'$ are convex and relatively closed subsets of $M_u$ and $\Theta_v$ respectively. We say that $\cE'$ is a \emph{mixed linear exponential family}, if both of $M_u'$ and $\Theta_v'$ are given by affine restrictions.}
\end{defn}

It is useful to introduce the following notation:
\begin{equation}\label{eq:cucv}
C_u\;=\;\{\bs\mu\in M:\;\bs\mu_u\in M_u'\},\qquad C_v\;=\;\{\bs\theta\in \Theta:\;\bs\theta_v\in \Theta_v'\}.
\end{equation}
The mixed convex exponential family is then given as an intersection $\cE'=\cE_u\cap\cE_v\subseteq \cE$, where $\cE_u=\{P_{\bs\mu}\cd  \bs\mu\in C_u\}$ and $\cE_v=\{P_{\bs\theta}\cd  \bs\theta\in C_v\}$ (one model is convex in the mean parametrization and the other is convex in the canonical parametrization). Note that these restrictions are also variation independent so that 
\[\cE'=\{P_{(\bs\mu_u,\bs\theta_v)}\cd\bs\mu\in C_u,\bs\theta\in C_v\}= \{P_{(\bs\mu_u,\bs\theta_v)}\cd(\bs\mu_u,\bs\theta_v)\in M'_u\times \Theta'_v\}.\]
%   \begin{enumerate}
%     \item [(i)] 
%     $\cE_u=\{P_{\bs\mu}\cd  \bs\mu\in C_u\}$ where $C_u\subseteq M$ is given by convex constraints on $\bs \mu_u$, i.e.\  $C_u=\{\bs\mu\in M\cd \bs\mu_u\in A_u\}$ for $A_u\subseteq \R^r$ convex and closed;
%     \item [(ii)] $\cE_v=\{P_{\bs\theta}\cd  \bs\theta\in C_v\}$ where $C_v\subseteq \Theta$ is given by convex constraints on $\bs \theta_v$, i.e.\  $C_v=\{\bs\theta\in \Theta\cd \bs\theta_v\in A_v\}$ for  $A_v\subseteq \R^s$ convex and closed.%; 
% \end{enumerate}
%   Then $\cE'$ is called a \emph{mixed convex} submodel of $\cE$ and a \emph{mixed convex exponential family}. We say that $\cE'$ is a \emph{mixed linear exponential family}, if both of $A_u$ and $A_v$ are affine.

We now discuss a few examples of this. Recall that the family of multivariate Gaussian distributions with zero mean and unknown covariance matrix $\Sigma$ is indeed a regular exponential family with inner product $\langle A,B\rangle=\tr(AB)$ and 
$$\bs t(\bs x) = -\bs x\bs x^T/2,\quad \bs\theta= K,\quad  \bs\mu=-\Sigma/2, \quad A(K)=-\frac 1 2 \log \det K.$$ 
The space of canonical parameters is the cone of positive definite matrices and the space of mean parameters is the cone of negative definite matrices. 
\begin{ex}\label{ex:pos-assoc}
Fix a graph $G=(V,E)$ on $d$ vertices $V=\{1,\ldots,d\}$ and consider the family of $d$-variate mean zero Gaussian distributions. We split the sufficient statistics into $\bs u=(-x_ix_j/2)_{ij\in E}$ and $\bs v=(-x_ix_j/2)_{ij\notin E}$. The diagonal entries $-x_i^2/2$ are included in $\bs u$. Then $\bs \mu_u=(-\Sigma_{ij}/2)_{ij\in E}$ and $\bs\theta_v=(K_{ij})_{ij\notin E}$. We may consider a mixed convex family given by  $\bs \mu_u\leq 0$ and $\bs\theta_v=0$. These are precisely the locally associated Gaussian graphical models discussed in Section~\ref{sec:laggm}.
\end{ex}

Mixed convex exponential families enable easy formulation of other relevant models encoding positive dependence in Gaussian distributions:
\begin{ex}\label{ex:pos-assoc2}
With the set-up as in Example~\ref{ex:pos-assoc},  we alternatively split the sufficient statistics  into $\bs u=(-x_i x_j/2)_{ij\notin E}$ and $\bs v=(-x_ix_j/2)_{ij\in E}$. The diagonal entries are now included in $\bs v$. Then $\bs \mu_u=(-\Sigma_{ij}/2)_{ij\notin E}$ and $\bs\theta_v=(K_{ij})_{ij\in E}$. We may consider a mixed convex family given by  $\bs \mu_u\leq 0$ and $\bs\theta_v\leq 0$. Here cliques in the graph correspond to subsystems characterized by a strong notion of positive dependence (all partial correlations nonnegative) or, in other words, the conditional distribution of variables in a clique given the remaining variables is $\mtp$. Otherwise the system is weakly positively dependent (positive correlations).
\end{ex}

An example of a mixed linear model can be easily motivated by causal analysis \cite{pearl1994can}, where zero restrictions on some entries of $\Sigma$ correspond to marginal independence and zero restrictions on $K$ correspond to conditional independence. Models of this form fit our general set-up:

\newcommand{\gcon}{G_{\textsc{con}}}
\newcommand{\gcov}{G_{\textsc{cov}}}
\begin{ex}\label{ex:wermuth}
Given a graph $\gcov$, called the covariance graph, over $V=\{1,\ldots,d\}$, we define the  corresponding covariance graph model $B(\gcov)$ given by all covariance matrices that satisfy $\Sigma_{ij}=0$ if $ij\notin E(\gcov)$. Given the covariance graph $\gcov$ and the concentration graph $\gcon$ we want to understand the intersection $B(\gcov)\cap M(\gcon)$. In the special case when $ij\notin E(\gcon)$ implies $ij\in E(\gcov)$, the corresponding model is a mixed linear model. Models of this type were studied in detail in \cite{kauermann:96}, see also \cite{boege2021geometry,drton2002new,pearl1994can}.
\end{ex}

\begin{ex}\label{ex:djordjilovic}
Recently \cite{djordjilovic2018searching} discussed the problem of testing equality of mean zero $d$-variate Gaussian distributions under the assumption that they lie in a fixed graphical model over the graph $G$. So suppose that $X,Y$ are independent Gaussian with covariance matrices $\Sigma^{(1)}, \Sigma^{(2)}$. Here the mixed parameters are $(\Sigma^{(1)}_{ij},\Sigma^{(2)}_{ij})_{ij\in E}$ for the mean part and $(K^{(1)}_{ij},K^{(2)}_{ij})_{ij\notin E}$ for the canonical part. The mixed linear model which assures that both distributions are equal is given by $\Sigma^{(1)}_{ij}=\Sigma^{(2)}_{ij}$ for all $ij\in E$ and $K^{(1)}_{ij}=K^{(2)}_{ij}=0$ for all $ij\notin E$. 
\end{ex}

Although the Gaussian case is our focus here, there are interesting examples beyond this case.

\begin{ex}Consider discrete random variables $X$ and $Y$ taking values in $\cS=\{0,1,\ldots, k\}$ and let $p_{xy}=P(X=x, Y=y)$ with $p_{xy}>0$ for all $x,y\in \cS$.  
This specifies a regular exponential family with canonical parameters 
$$\lambda_{xy}, \quad x,y\in \cS \text{ and } (x,y)\neq(0,0)$$ where
$$\lambda_{xy}=\log p_{xy}-\log p_{00}$$
and so $\lambda_{00}=0$. The corresponding sufficient statistics are $\bs t=(t_{xy})$ where
$$t_{xy}(x',y')=\mathbf{1}_{\{(x,y)\}}(x',y')$$ and the corresponding mean parameters are $p_{xy}$ for $(x,y)\neq (0,0)$. We will consider a linear transformation of the canonical parameters and sufficient statistics into
$$\theta_{xy}=\lambda_{xy}-\lambda_{x0}-\lambda_{0y},\; x,y\in \cS\setminus \{0\},
\quad \theta_{x0}=\lambda_{x0},\; \theta_{0y}=\lambda_{0y},\; x,y\in \cS\setminus \{0\}$$
with corresponding sufficient statistics 
$$t_{xy},\; x,y\in \cS\setminus\{0\}; \quad t_{x+}= \sum_{y\in \cS}t_{xy},\; x\in \cS\setminus\{0\}\quad t_{+y}=\sum_{x\in \cS}, \; y\in \cS\setminus\{0\}.$$
This exponential family may thus be mixed parametrized with the \emph{marginals}
\begin{gather*}\mu_{x+}=p_{x+}= \E\{t_{x+}(X,Y)\}=\sum_{y}p_{xy},\;x \in \cS\setminus\{0\},\\
\mu_{+y}=p_{+y}=\E\{t_{+y}(X,Y)\}=\sum_{x}p_{xy},\;y \in \cS\setminus\{0\}\end{gather*}
and the \emph{interactions}
$$\theta_{xy},\;x,y\in \cS\setminus\{0\}.$$
We may then consider the hypothesis of \emph{marginal homogeneity} (\cite{kullback:71}), i.e.\ 
\begin{equation}\label{eq:marghom}p_{x+}=p_{+ x} \text{ for all $x\in \cS$}\end{equation} in combination with the distribution being $\mtp$; the latter is equivalent to the condition 
\begin{equation}\label{eq:mtp2disc}\theta_{xy}+\theta_{x'y'}-\theta_{xy'}-\theta_{x'y}\geq 0 \text{ for all $x\geq x'$ and $y\geq y'$}.\end{equation}
The restriction \eqref{eq:marghom}  is convex (in fact linear) in the mean parameters and \eqref{eq:mtp2disc} is convex in $\theta_{xy}$
so these restrictions jointly specify a mixed convex exponential family.  

Another alternative would exploit that categories are ordered and for example specify that $p_{i+}$ is stochastically smaller than $p_{+i}$ i.e.\ 
$$\sum_{x=0}^{j}p_{x+}\leq\sum_{y=0}^jp_{+y} \text{ for all $j\in \cS$},$$ 
yielding a convex restriction also on the mean parameters; see \cite{agresti:83} and \cite{agresti2003categorical} for further details of this model. 
\end{ex}

The mixed parametrization can be naturally used in discrete exponential families when the mean vector is regressed via a link function on some covariates and the remaining canonical parameters are used to handle higher order interactions. Similar ideas emerge for models with restrictions on marginal and conditional distributions; see e.g.\ \cite{fitzmaurice1993likelihood} and \cite{glonek1996class}.

\section{Estimation in mixed convex exponential families}\label{sec:estimation}

Since mixed convex models are not necessarily convex exponential families (given by convex constraints on $\bs\theta$ only),  maximum likelihood estimation leads in general to non-convex optimization problems that may have many local optima. In this section we propose a simple alternative approach leading to \emph{two} convex optimization problems. 
In Section~\ref{sec:asym} we will show that, asymptotically, the resulting estimator has the same asymptotic distribution as the MLE up to the first order in the sense that the difference between the estimators converges in probability to zero, even after multiplying with $\sqrt{n}$.

\subsection{Likelihood and its dual}

Before we present our optimization procedure, we quickly recall the definition of the dual likelihood function; cf.\ Chapter~6 in \cite{brown:86}. Given a random sample $X^{(1)},\ldots,X^{(n)}$ of size $n$ from the exponential family $\cE$ in \eqref{eq:expfam}, denote $\bs t=\bs t_n=\sum_{i=1}^n \bs t(X^{(i)})/n$. The log-likelihood function is a strictly concave function of $\bs \theta$ given as
\begin{equation}\label{eq:likelihood}
\ell(\bs \theta;\bs t)\;=\;\<\bs \theta,\bs t\>-A(\bs \theta),
\end{equation}
where we here and in the following have suppressed the index $n$ as we are not yet considering the asymptotic behaviour.  Since 
$$\nabla \ell(\bs \theta;\bs t)=\bs t-\nabla A(\bs \theta)=\bs t-\mu(\bs \theta),$$
the unconstrained optimum based on data $\bs t$ is the parameter $\bs \theta$ for which the mean parameter $\mu(\bs \theta)$ is equal to $\bs t$; in other words, this is $\theta(\bs t)$ and is well-defined if $\bs t\in M$. In what follows we ignore that $\bs t$ comes from data and write $\ell(\bs \theta;\bs\mu)$, where $\bs \mu$ is a general point in the topological closure $\overline M$ of the space of mean parameters.

%\piotr{I rewrote the following paragraphs because the fact that $\bs\theta$ and $\bs\mu$ were correspnding to two different points was confusing.}
The \emph{Fenchel conjugate} of the cumulant function $A$ is the function 
$$A^*(\bs \mu)\;=\;\sup\{\ell(\overline{\bs \theta};\bs \mu): \overline{\bs\theta}\in \R^k\}.$$  
The function $A^*$ is convex as a supremum of linear functions and indeed strictly convex. 
The unique optimizer of the log-likelihood is $\theta(\bs \mu)$, so
\begin{equation}\label{eq:conjugate}
A^*(\bs \mu)\;=\;\ell(\theta(\bs \mu);\bs \mu)\;=\;\<\theta(\bs \mu),\bs \mu\>-A(\theta(\bs \mu))
\end{equation}
implying in particular that $A^*$ is smooth, i.e.\ infinitely often differentiable, since $\theta$ and $A$ are both smooth.
For any fixed $\overline{\bs \theta}$, the function
\begin{equation}\label{eq:duallike}
\widecheck \ell(\bs \mu;\overline{\bs \theta})\;:=\;\<\overline{\bs \theta},\bs \mu\>-A^*(\bs \mu)
\end{equation}
is a strictly concave function of $\bs \mu\in M$ called the \emph{dual log-likelihood function}. Analogously to the log-likelihood function, $\widecheck\ell$ satisfies 
$$\nabla_{\bs\mu} \widecheck \ell(\bs \mu;\overline{\bs\theta})\;=\;\overline{\bs \theta}-\nabla_{\bs\mu} A^*(\bs \mu)\;=\;\overline{\bs \theta}-\theta(\bs \mu),$$ which follows by composite differentiation in \eqref{eq:conjugate} since $\mu(
\theta(\bs \mu))=\bs \mu$ and thus 
$$\nabla_{\bs\mu} A^*(\bs \mu)=\theta(\bs\mu) +\nabla_{\bs \mu}\theta(\bs \mu)\cdot \bs \mu -\nabla_{\bs\mu}\theta(\bs\mu)\cdot \mu(
\theta(\bs \mu))=\theta(\bs\mu).$$

Consider two distributions in $\cE$, one with the mean parameter $\bs \mu_1$ and the other with canonical parameter $\bs \theta_2$. The \emph{Kullback--Leibler divergence} (\cite{kullback:leibler:51}) between these two distributions is
\begin{equation}\label{eq:KL}
\K(\bs \mu_1,\bs \theta_2)\;=\;-\<\bs \mu_1,\bs \theta_2\>+A^*(\bs\mu_1)+A(\bs \theta_2),
\end{equation}
which follows directly from Proposition 6.3 in \cite{brown:86} and the definition of $A^*(\bs \mu_1)$.
The Kullback--Leibler divergence is well defined and nonnegative over $M\times \Theta$. Moreover, $\bs K(\bs \mu_1,\bs\theta_2)=0$ if and only if $\bs\mu_1=\mu(\bs\theta_2)$. We will extend the definition of $\K(\cdot,\bs \theta
)$ to all $\R^k$ by semicontinuity; cf.\ \cite[p.\ 175]{brown:86}. Then $\K(\bs \mu,\bs \theta)$ may be well-defined even if $\bs \mu$ does not lie in the space of mean parameters but in its closure.

The reason to express the Kullback--Leibler distance in terms of $\bs \mu_1$ and $\bs \theta_2$ rather than $\bs \theta_1,\bs \theta_2$ (as usually done in the literature) is that we wish to exploit the following basic result.
\begin{prop}\label{prop:KLconvex}
The Kullback--Leibler divergence $\K(\bs \mu_1,\bs \theta_2)$ is strictly convex both in $\bs \mu_1$ and in $\bs \theta_2$.
\end{prop}
\begin{proof}This follows dicrectly from \eqref{eq:KL} and the fact that both $A(\bs \theta)$ and $A^*(\bs \mu)$ are strictly convex functions
\end{proof}
%\textcolor{purple}{
\begin{rem}\label{rem:kulllike}
Minimizing $\K(\bs t,\bs \theta)$ with respect to $\theta$ with $\bs t$ fixed is equivalent to maximizing $\ell(\bs\theta;\bs t)$ in \eqref{eq:likelihood}. Similarly, minimizing $\K(\bs\mu,\overline{\bs\theta})$ with respect to $\bs \mu$ with $\overline{\bs\theta}$ fixed is equivalent to maximizing the dual log-likelihood $\widecheck{\ell}(\bs\mu;\overline{\bs\theta})$ in \eqref{eq:duallike}.
\end{rem} 

\subsection{The mixed dual estimator}\label{sec:MDE}

Recall the definition of the sets $C_u, C_v$ in \eqref{eq:cucv} and that $\bs t=\sum_{i=1}^n \bs t(X^{(i)})/n$. We propose the following two-step procedure to estimate the mixed parameter $(\bs\mu_u,\bs\theta_v)$ in the mixed convex family $\cE'$ from data $\bs t$:
\begin{enumerate}
    \item [(S1)] First minimize $\K(\bs t,\bs\theta)$ over $\bs \theta \in C_v\subseteq \Theta$. Denote the \emph{unique} optimum, assuming it exists, by $\widehat{\bs \theta}$. 
    \item [(S2)] Then minimize $\K(\bs \mu,\widehat{\bs\theta})$ subject to $\bs \mu\in C_u\subseteq M$. Denote the unique optimum by $\widecheck{\bs \mu}$. 
\end{enumerate}
\noindent We shall term our estimator $\widecheck{\bs \mu}$ the \emph{mixed dual estimator} (MDE) {and show below in Theorem~\ref{th:identity} that indeed $P_{\widecheck{\bs\mu}}\in \cE'$.}
By Proposition~\ref{prop:KLconvex}, both steps (S1) and (S2) rely on solving a convex optimization problem. Note that the optimum in (S1) is the MLE under the convex exponential family given by  $\bs\theta \in C_v$ (cf.\ Remark~\ref{rem:kulllike}). This MLE may not exist (if $\bs t$ lies on the boundary of the space of mean parameters) but we have the following:
\begin{prop}\label{prop:s1thens2}
If the optimum $\widehat{\bs \theta}$ in (S1) exists then it is unique and the optimum $\widecheck{\bs \mu}$ in (S2) exists and is unique too.
\end{prop}
\begin{proof}
Suppose that the optimum in (S1) exists. Uniqueness follows by strict convexity. Note that (S2) is equivalent to the maximization of the dual likelihood $\widecheck\ell(\bs \mu,\widehat{\bs\theta})$ over all $\bs\mu\in  C_u$. Let $\widecheck M=\{\bs\mu:A^*(\bs \mu)<\infty\}$, which contains $M$ but is typically bigger, $M\subseteq \widecheck M\subseteq\overline M$ where $\overline M$ is the {topological}  closure of $M$. Let $S$ be the closure (in $\widecheck M$) of $C_u$. Since $C_u$ is relatively closed in $M$, the only extra points are those in $\widecheck M\setminus M$. Now, by Theorem~6.13 in \cite{brown:86}, the fact that $\widehat{\bs \theta}\in \Theta$ implies that $\widecheck\ell(\bs \mu,\widehat{\bs\theta})$ attains an optimum over $S$ and the optimum is uniquely attained in $S\cap M=C_u$. In particular, the optimum in (S2) exists. (Note that the extra points we added to the model played only an auxiliary role in this proof.)
\end{proof}

\begin{rem}
Let $\bs t=(\bs u,\bs v)$. Theorem~\ref{th:identity} below implies that $\bs u\in M_u$ is a necessary condition for $\widehat{\bs \theta}$ to exist. 
\end{rem}
Our main result of this section shows that after the steps (S1) and (S2), we indeed obtain a point in the mixed convex family $\cE'$. 

\begin{thm}\label{th:identity}Let $\bs t=(\bs u,\bs v)$ and suppose that $\widehat{\bs \theta}$ in step (S1) exists. Then, $\mu_u(\widehat{\bs \theta})=\bs u\in M_u$ and in step (S2) we get that $\theta_v(\widecheck{\bs \mu})=\widehat{\bs \theta}_v$. In particular, after steps (S1) and (S2), the optimum $\widecheck{\bs \mu}$ represents an element in the mixed convex family $\cE'$. 
\end{thm}
%\steffen{Edit and expand this proof}
% \begin{proof}
% In Step (S1) we minimize $\K(\bs t;\bs \theta)$ with respect to $\bs \theta\in C_v$. Using the mixed parametrization we can first minimize with respect to $\mu_u(\bs \theta)$ and then with respect to $\bs\theta_v$. Since no restrictions on $\mu_u(\bs \theta)$ are imposed on $\mu_u(\bs \theta)$, it will be equal to its unrestricted optimum $\bs u$. In particular, we need that $\bs u\in M_u$. The same argument applies to Step (S1), where we minimize $\K(\bs \mu;\widehat{\bs \theta})$ with respect to $\bs \mu\in C_u$. In the mixed parametrization we optimize with respect to the vector $(\bs \mu_u,\theta_v(\bs \mu))$. Since there are no restrictions on $\theta_v(\bs \mu)$, at the optimum $\widecheck{\bs \mu}$ it will be equal to the unrestricted optimum $\widehat{\bs\theta}_v$. This implies that $\theta(\widecheck{\bs \mu})\in C_v$.
% \end{proof}
%\color{purple}
\begin{proof}
If the optimum $\widehat{\bs\theta}\in C_v$ in step (S1) exists then, by convexity of $\K(\bs t,\bs\theta)$ and of $C_v$, it must satisfy $\<\nabla_{\bs \theta}\K(\bs t,\widehat{\bs\theta}),\bs\theta-\widehat{\bs\theta}\>\geq 0$ for all $\bs\theta\in C_u$. Since $\Theta$ is an open set, a small perturbation $\widehat{\bs\theta}+\bs \tau$ also lies in $\Theta$. If, in addition, $\bs\tau_v=\bs 0$ then this perturbation lies in $C_v$. Thus, for any sufficiently small vector $\bs\tau\in \R^k$ such that $\bs\tau_v=\bs 0$, we must have 
$$
\<\nabla_{\bs \theta}\K(\bs t,\widehat{\bs\theta}),\bs\tau\>=\<\mu_u(\widehat{\bs\theta})-\bs u,\bs\tau_u\>\geq 0.
$$
Since $\bs\tau_u$ is small but otherwise arbitrary, this is only possible if $\mu_u(\widehat{\bs\theta})=\bs u$, which proves the first part of the theorem. 

The second part is proved in the same way: In step (S2), the optimum $\widecheck{\bs \mu}$ exists and is unique by Proposition~\ref{prop:s1thens2}. By convexity of $\K(\bs \mu,\widehat{\bs \theta})$, the optimum satisfies $\<\nabla_{\bs\mu}\K(\widecheck{\bs \mu},\widehat{\bs\theta}),\bs\mu-\widecheck{\bs\mu}\>\geq 0$ for all $\bs\mu\in C_u$. Since $\widecheck{\bs\mu}$ is an interior point of $M$, a small perturbation $\widecheck{\bs\mu}+\bs\tau$ also lies in $M$ and, if $\bs\tau_u=\bs 0$, it also lies in $C_u$. For any such perturbation, we necessarily have
$$
\<\nabla_{\bs \mu}\K(\widecheck{\bs \mu},\widehat{\bs\theta}),\bs\tau\>=\<\theta_v(\widecheck{\bs\mu})-\widehat{\bs \theta}_v,\bs\tau_v\>\geq 0.
$$
Since $\bs\tau_v$ is small but otherwise arbitrary, this is only possible if $\theta_v(\widecheck{\bs\mu})=\widehat{\bs \theta}_v$.
\end{proof}
%\steffen{Notation needs to be taken care of: $C_u$ vs.\ $M'_u$ etc.}\piotr{I tried to remove $C_u$ but writing all the time ``the set of all $\bs\mu$ such that $\bs\mu_u\in M_u'$'' is hard. When we work in the mixed parametrization $M_u'$ and $\Theta_v'$ is also convenient to have. We can add a remark in which we explain why we keep both...}
\color{black}

\begin{rem}
In principle, we may interchange the order of optimization with respect to $\bs\theta$ and $\bs \mu$  by starting from $\theta(\bs t)$ and running Step (S2) first. However, if $\bs t$ does not lie in the space of mean parameters (but in its closure) $\theta(\bs t)$ is not well-defined. Here we exploit that $\widehat{\bs\theta}$ might exist even if $\bs t$ is on the boundary of $M$ due to the additional convex restriction $\bs\theta\in C_v$. 
\end{rem}

\begin{rem}Note that when $C_v$ is given by affine constraints, the first step just corresponds to reducing data by sufficiency; then the MDE is simply the \emph{dual likelihood estimate} (DLE), studied by \cite{brown:86,christensen:89,efron:78} and used extensively by \cite{kullback:59}. This provides, for example,  a straight-forward way to test equality of distributions in Example~\ref{ex:djordjilovic} because the likelihood ratio statistics based on the dual likelihood has asymptotically the same distribution as the standard likelihood ratio statistics, cf.\ \cite[Theorem 3.3]{christensen:89}. 
\end{rem}

\section{Estimating Gaussian locally associated distributions}\label{sec:EP}
As mentioned in Section~\ref{sec:laggm}, a  locally associated Gaussian graphical model $\mpl(G)$ is determined as
$$\mpl(G)\;\;=\;\;\ep(G)\cap M(G)$$
and thus it forms an instance of a mixed convex exponential family since $\ep(G)$ is convex in $\Sigma$ and $M(G)$ is convex in $K$ and the restrictions refer to distinct parts of the canonical statistic and parameters; see Example~\ref{ex:pos-assoc} for details.

\subsection{The Gaussian log-likelihood}\label{sec:like}

Given the data $\bs X\in \R^{n\times d}$ with independent rows distributed as $N(0,\Sigma^*)$ our goal is to estimate $\Sigma^*$. Let  $S=\bs X^T\bs X/n$ be the sample covariance matrix. The Gaussian log-likelihood (ignoring a multiplicative constant $n$) is  
\begin{equation}\label{eq:gausllike}
\ell(K)\;\;=\;\; \frac{1}{2}\log\det(K)-\frac{1}{2}\tr(S K).
\end{equation}

Note that here and elsewhere we ignore the multiplicative constant $n$ in the log-likelihood function. 
Recall that for the Gaussian family  $$\bs t(\bs x) = -\bs x\bs x^T/2,\quad \bs\theta= K,\quad  \bs\mu=-\Sigma/2, \quad A(K)=-\frac 1 2 \log \det K.$$
The unique optimizer of $\ell(K)$ is $S^{-1}$ and so 
$$
A^*(-\Sigma/2)\;=\;-\frac{1}{2}\log\det(\Sigma)-\frac{d}{2}
$$
and thus the Kullback--Leibler divergence as above becomes
\[\K(-\Sigma_1/2,K_2) \;=\; \frac12\tr(\Sigma_1K_2-I)-\frac12\log \det (\Sigma_1K_2).\]

\subsection{The mixed dual estimator}\label{sec:mde}
We consider now the  mixed dual estimator MDE as developed above - in this case equivalent to the dual likelihood estimator, since the first step is just estimation in the standard graphical model given by edge restrictions on $K$ (see e.g.\ \cite[Section~5.2.1]{lau96}). %

The second step  corresponds to the following convex optimization problem:
\begin{align}\label{eq:dualMLE}
& \underset{\Sigma \succeq 0}
{\text{minimize}}
& & -\log\det(\Sigma)+\tr(\Sigma\widehat{K}) \\
\nonumber & \text{subject to}
&& \Sigma_{ij}\geq 0 \text{ for all } ij\in E(G),
\end{align}
where, as before, $\widehat{K}$ is the MLE of $K$ in the linear exponential family $M(G)$.  Note that we ignore the Markov constraint on $\Sigma$ that would destroy convexity of this problem. As shown in Theorem~\ref{th:identity} above, the Markov constraint becomes automatically fulfilled in the optimization. 
The Lagrangian for this problem is
\begin{eqnarray*}L(\Sigma,\Lambda)=-\log\det \Sigma + \tr(\Sigma\widehat{K})- \tr(\Sigma\Lambda)=
-\log\det \Sigma + \tr(\Sigma(\widehat{K} - \Lambda))\end{eqnarray*}
where $\Lambda$ is a symmetric matrix with diagonal equal to $0$, $\Lambda_{ij}= 0$ unless $ij\in E(G)$, and  $\Lambda_{ij}\geq 0$.
The Lagrangian is mimimized in $\Sigma$ for fixed $\Lambda$ by the  matrix $\Sigma^*$ given as
\begin{equation}\label{eq:stat_lang}\Sigma^*
= (\widehat{K} -\Lambda)^{-1}\end{equation}
and complementary slackness implies for the optimal $\Lambda=\Lambda^*$ that
\[\Sigma^*_{ij}\Lambda^*_{ij}=0,\quad ij\in E\] or, equivalently,
\[\Sigma^*_{ij}(\widehat{K}_{ij}- K^*_{ij})=0,\quad ij\in E.\] 
Note that it follows directly from \eqref{eq:stat_lang} that $\Sigma^*\in M(G)$  if $\widehat{\Sigma}\in M(G)$, conforming with Theorem~\ref{th:identity} and hence $\Sigma^*=\widecheck{\Sigma}$ is in fact the MDE we are looking for.
In summary, we have the following:
\begin{thm}\label{thm:dualMLE}If the MDE $\widecheck{\Sigma}$ of $\Sigma$ under $\mpl(G)$  exists, it is given as the unique positive definite solution to the following system of equations and inequalities, where $\widecheck K=\widecheck{\Sigma}^{-1}$ and $\widehat{K}
=\widehat{\Sigma}^{-1}$:
\begin{enumerate}
\item [i)] $\widecheck{\Sigma}_{ij} \geq 0, \quad ij\in E(G)$
\item [ii)] $\widehat{\Sigma}_{ij}=S_{ij}, \quad ij\in E(G)$
\item [iii)] $\widehat{\Sigma}_{ii}=S_{ii}, \quad i\in V(G)$
\item [iv)] $\widecheck K_{ij} = \widehat{K}_{ij}=0, \quad ij\notin E(G)$
\item [v)] $\widecheck K_{ij} \leq \widehat{K}_{ij}, \quad ij\in E(G)$
\item [vi)] $\widecheck K_{ii} =\widehat{K}_{ii}, \quad i\in V(G)$
\item [vii)] $\widecheck{\Sigma}_{ij}(\widehat{K}_{ij}- \widecheck K_{ij})=0,\quad  ij\in E(G)$.
\end{enumerate}
\end{thm}
Note that the equations (ii), (iii), and the last part of (iv) are equations determining the MLE $\widehat{K}$ under $M(G)$. The condition (vii) naturally induces sparsity in $\widecheck{\Sigma}$.

Further, denoting by $\widecheck G$ the graph whose edges correspond to non-zero entries of $\widecheck{\Sigma}$, we note that the zero entries in $\Sigma$ obtained in this way are complementary to the zero entries in $K$; cf.\ also Example~\ref{ex:wermuth}. We also have the following.

\begin{cor}\label{cor:linear_version}The mixed dual estimate $\widecheck{\Sigma}$  in the model $B(\widecheck G) \cap M(G)$ is identical to the mixed dual estimate in the model $\mpl(G)=A(G)\cap M(G)$  as determined in Theorem~\ref{thm:dualMLE} above.
\end{cor}

\begin{proof}
By construction, if $ij\notin E(\widecheck G)$ then $ij\in E(G)$. Denoting by $E_0$ the set of pairs that do not lie in $E(\widecheck G)$ we get that $E_0\subseteq E(G)$. The optimality conditions for $\widehat{K}$ over $M(G)$ are standard and given by
\begin{enumerate}
\item [(a)] $\widehat{\Sigma}_{ii}=S_{ii}, \quad i\in V(G)$
\item [(b)] $\widehat{\Sigma}_{ij}=S_{ij}, \quad ij\in E(G)$
\item [(c)] $\widehat{K}_{ij}=0, \quad ij\notin E(G)$.
 \end{enumerate}
By an analogous argument, the dual likelihood estimate in $B(\widecheck G)$ based on $\widehat{K}$ is the unique positive definite matrix $\widetilde \Sigma=\widetilde K^{-1}$ satisfying
\begin{enumerate}
\item  [(d)] $\widetilde K_{ii}=\widehat{K}_{ii},\quad i \in V(G)$
\item  [(e)] $\widetilde K_{ij}=\widehat{K}_{ij}, \quad ij\in E(\widecheck G)$
\item  [(f)] $\widetilde \Sigma_{ij}=0,\quad ij\in E_0$.
 \end{enumerate}
Our aim is to show that $\widetilde \Sigma=\widecheck{\Sigma}$. First note  that, (d) together with  condition (vi) in Theorem~\ref{thm:dualMLE}, implies that $\widetilde K_{ii}=\widecheck K_{ii}$ for all $i\in V$. Similarly, (iv) and (e) imply that $\widetilde K_{ij}=\widehat  K_{ij}=\widecheck K_{ij}$ for $ij\notin E(G)$. This equality extends to all $ij\in E(G)\setminus E_0$ by (vii) and (e). Moreover $\widecheck{\Sigma}_{ij}=\widetilde\Sigma_{ij}=0$ for all $ij\in E_0$. It is a standard result that there is a unique completion of a partially specified positive definite matrix $K$ to a matrix such that $\Sigma=K^{-1}$ has zero entries on the unspecified entries of $K$. It follows that  $\widetilde K=\widecheck K$.
\end{proof}
Note that the statement in Corollary~\ref{cor:linear_version} is not trivial since $B(\widecheck G) \cap M(G)$ is not a subset of $\mpl(G)$.

The conditions for existence of $\widecheck{\Sigma}$ do not seem to simplify beyond the conditions for existence of $\widehat{K}$; see also Theorem~\ref{th:identity}.  Conditions for existence of $\widehat{K}$ may be rather involved, see for example \cite{uhler:12,gross:sullivant:18}.

\section{Asymptotic behaviour of the mixed dual estimator}\label{sec:asym}

%\piotr{I now used the mixed parametrization more prominently below. }
We now return to the the general mixed dual estimator. In this section we work entirely in the corresponding mixed parametrization $\bs\psi=(\bs\mu_u,\bs\theta_v)$. The MLE and the MDE for $\cE'$ in this parametrization are denoted simply by $\widetilde{\bs \psi}$ and $\widecheck{\bs \psi}$. The maximizer of the log-likelihood function obtained in step (S1) is denoted  $\widehat{\bs\psi}$. Suppose that the true data generating distribution with the mixed parameter $\bs \psi_0$ lies in $\cE'$, that is, $\bs\psi_0\in M_u'\times \Theta'_v$. We study the asymptotic distribution of the MDE $\widecheck{\bs \psi}_n$, where $n$ is the sample size and show that the MDE  is $\sqrt{n}$-consistent and has exactly the same asymptotic distribution as the maximum likelihood estimator $\widetilde{\bs \psi}_n$. 

\begin{thm}\label{th:equal}
The MDE and MLE are asymptotically equivalent, i.e.\ it holds that $\sqrt{n}(\widecheck{\bs\psi}_n-\widetilde{\bs\psi}_n)\to 0$ in probability, implying that $\sqrt{n}(\widecheck{\bs \psi}_n-\bs \psi_0)$ converges in distribution  and the limiting distribution of $\sqrt{n}(\widecheck{\bs \psi}_n-\bs \psi_0)$ equals the limiting distribution of $\sqrt{n}(\widetilde{\bs\psi}_n-\bs \psi_0)$.
\end{thm}
\begin{proof} The proof is provided in Appendix~\ref{app:proof}.
\end{proof}

Let $\overline{\bs\psi}_n=(\bs u_n, \theta_v(\bs t_n))$ be the {MLE in the unrestricted exponential family $\cE$ expressed in the mixed parametrization.} The limiting distribution in Theorem~\ref{th:equal} is obtained by projecting the Gaussian limiting distribution of $\sqrt{n}(\overline{\bs \psi}_n-\bs \psi_0)$ onto the tangent cone of the mixed exponential family at the true parameter $\bs \psi_0$; cf.\ \cite{geyer1994asymptotics}. If the constraints defining $\Theta_v'$ are affine, it is useful to equivalently describe this distribution as the limiting distribution of $\sqrt{n}(\widehat{\bs \psi}_n-\bs \psi_0)$ (which is also Gaussian), onto the tangent cone $T_{M_u'\times \Theta_v'}(\bs \psi_0)$. In the case of locally associated Gaussian graphical models, this results with the Gaussian distribution (the asymptotic distribution of the MLE in a Gaussian graphical model; cf.\ \cite[Proposition~5.8]{lau96}) projected onto the cone given by the edge covariances being nonnegative. The problem is that even if we can describe this distribution exactly, it depends in a very complicated way on the true covariance matrix. It will be given as a mixture of normal distributions  that are `truncated' to regions projecting onto the various facets with weights that are generally impossible to compute exactly. As an alternative to using asymptotic results, the distribution of the MDE may be simulated using bootstrap methods, as the estimation algorithm is fast and guaranteed to converge, whereas simulating the distribution of the MLE is difficult as the MLE may not even  be well-defined for all bootstrap samples.

\section{Learning the local structure}\label{sec:pglasso}

In this section we consider the situation where the graph $G$ defining the local structure in locally associated Gaussian graphical models is unknown. We aim at obtaining a sparse structure in $K$ through a lasso type penalty.
\subsection{The positive graphical lasso}
To avoid losing any positive dependence information we only penalize \emph{positive} values in the inverse covariance matrix, corresponding to large negative partial correlations. More precisely, we want to solve  the following optimization problem
\begin{equation}\label{eq:pospenal}
{\rm minimize}\;\;\;-2\ell(K)+\rho\|K^+\|_1\qquad\mbox{subject to } K\succ 0,
\end{equation}
where $\ell_n(K)$ is the Gaussian log-likelihood in \eqref{eq:gausllike} and $$\|K^+\|_1=\sum_{i\neq j} \max\{0,K_{ij}\}.$$ We shall refer to this procedure as the \emph{positive graphical lasso}. Note that for $\rho=\infty$, the penalty forces the solution $\widehat{K}^\rho$ to be an M-matrix and hence  the positive graphical lasso can be seen as a direct relaxation of the estimation under the assumption that the distribution  is $\mtp$ (\cite{LUZ,slawski2015estimation}); cf.\ Section~\ref{sec:m-matrix}.

By computing subgradients, we easily check that $\widehat K^\rho$ is the unique optimal point of \eqref{eq:pospenal} if and only if 
\begin{equation}\label{eq:0rho}
\widehat \Sigma^\rho_{ij}-S_{ij}\;\in\;\begin{cases}
\{0\} & \mbox{if } \widehat K_{ij}^\rho<0\\
[0,\rho] & \mbox{if } \widehat K_{ij}^\rho=0\\
\{\rho\} & \mbox{if } \widehat K_{ij}^\rho>0\\
\end{cases}\qquad\mbox{for all }{i\neq j}\in V,
\end{equation}As a corollary we get an alternative characterization of the optimal solution that links it to the MLE in the Gaussian graphical model. 
\begin{cor}Let $\widehat G^\rho=(V,\widehat E^\rho)$ be the graph with edges determined by $\widehat K^{\rho}$ and define the modified sample covariance as
\[
S^\rho_{ij}\;=\;\begin{cases}
S_{ij} & \mbox{if } \widehat K_{ij}^\rho\leq 0\\
S_{ij}+\rho & \mbox{if } \widehat K_{ij}^\rho>0\\
\end{cases}\qquad\mbox{for all }i\neq j\in V.
\]
Then $\widehat K^\rho$ is the MLE under the Gaussian graphical model determined by $\widehat G^\rho$, based on the modified sample covariance $S^\rho$.
\end{cor}
\begin{proof} The MLE is uniquely determined by fitting covariances along edges in $\widehat E^\rho$ and satisfying $\widehat K_{ij}=0$ for non-edges. 
\end{proof}
See also Proposition~\ref{prop:kkt} and Corollary~\ref{cor:optimum} below for generalizations of this result.

The positive graphical lasso estimate, as described and calculated above, will avoid large negative partial correlations and as such it may often directly result in a locally associated covariance matrix, in particular for large penalty parameters, as shown in Theorem~\ref{th:InfgivesMLE}. If this is not the case, 
we may  wish to take the additional restriction  of edge positivity into account using the estimate $\widecheck{\Sigma}^\rho$ given as
\[\widecheck{\Sigma}^\rho = \argmin_{\Sigma\in \mpl(\widehat G^\rho)} \K(-\Sigma/2,\widehat{K}^\rho)=\argmin_{\Sigma\in \ep(\widehat G^\rho)} \K(-\Sigma/2,\widehat{K}^\rho).\]
This is exactly the dual likelihood estimate in \eqref{eq:dualMLE} calculated with $\widehat{K}^\rho$ as starting point, rather than $\widehat{K}$.  
We may then again think of the two-step procedure as first obtaining a compact representation $\widehat{K}^\rho$ of the data matrix $S$, adapting and taking into account the penalty for negative partial correlations, and subsequently approximating this compact representation with a locally associated, and hence locally associated Gaussian distribution $\widecheck{\Sigma}^\rho$. We refer to this procedure as the \emph{dual penalized likelihood estimate} (DPLE).
\subsection{A comment on high-dimensional analysis}\label{sec:highdim}

A careful analysis of the high-dimensional properties of the positive graphical lasso estimator is beyond the scope of this paper. Here we share some cautionary remarks. In \cite{soloff2020covariance} the authors analysed the convergence rates for the operator norm $\|\widehat K-K^*\|$ for the problem of M-matrix estimation. As we argue in Section~\ref{sec:m-matrix} below, M-matrix estimation is a special case of our positive graphical LASSO set-up. In particular, Section~3 in \cite{soloff2020covariance} suggests that we cannot expect good rates for the operator norm $\|\widehat K^\rho-K^*\|$ if $K^*$ is sparse. So in the high-dimensional regime, the spectral properties $\widehat\Sigma^\rho$ should be interpreted with caution. Similarly, $\widehat K$ will not have good support recovery. For example, if $K^*=I_d$, $\widehat K$ will typically not even be sparse. A natural way to obtain an estimator with better statistical properties is by replacing the sample covariance matrix $S$ in \eqref{eq:pospenal} with a better estimator of $\Sigma^*$ (e.g. shrinkage estimator). Our motivating example is high-dimensional with $d=136$ and $n=43$. In this case the positive graphical lasso estimator outperforms the graphical lasso estimator by a large margin also producing a much sparser graph; cf.\ Section~\ref{sec:gene}.

 \section{The GOLAZO algorithm}\label{app:golazo}

In this section we formulate a general optimization problem and algorithm that unifies the positive graphical lasso, estimation under $\mtp$ constraints, under local association, and a number of other forms of likelihood-based estimation for graphical models.  It allows us to flexibly introduce sign constraints, zero restrictions, and to penalize different signs of $K_{ij}$ at different rates.  
 
Let $L,U$ be two symmetric matrices with entries in $\R\cup \{-\infty,+\infty\}$ with the restriction that $L_{ij}\leq 0\leq U_{ij}$ for all $i\neq j$, and $L_{ii}=U_{ii}=0$ for all $i$. Denote 
$$
\|K\|_{LU}\;:=\;\sum_{i\neq j}\max\{L_{ij}K_{ij},U_{ij}K_{ij}\}.
$$
%\piotr{The way we define it here, we do not penalize the diagonal. This was not consistent with the rest of the text. But now I rewrote the rest to match this - penalizing the diagonal is not so interesting anyway.}
The function $\|K\|_{LU}$ is convex, positively homogeneous, continuous, and non-negative. Although it is sublinear, that is $\|K+K'\|_{LU}\leq \|K\|_{LU}+\|K'\|_{LU}$, it does not define a norm unless $|L_{ij}|=|U_{ij}|$ for all $i\neq j$. 

The penalty function $\|K\|_{LU}$ enables us to obtain sparse estimates of $K$ in a way that takes into account the signs of $K$ or, equivalently, the signs of the partial correlations. We aim at solving the following problem
    \begin{equation}\label{eq:pospenalR}
{\rm minimize}\;\;\;-2\ell(K)+\|K\|_{LU},
\end{equation}
and refer to this as \emph{Graphical Oriented LAZy Optimization} (\emph{GOLAZO}). To get a procedure that is invariant under diagonal rescaling we also typically replace the sample covariance matrix $S$ in $\ell_n(K)$ with the sample correlation matrix $R$. 

\begin{rem}
For non-paranormal distributions we replace the sample correlation matrix $R$ with another estimate of the correlation matrix. Following the SKEPTIC approach of \cite{liu2012high}, we first compute Kendall's tau $\widehat  \tau_{ij}$, which can be estimated without knowledge of the underlying monotone transformations $f_i$. Then we compute $\widehat \rho_{ij}=\sin(\tfrac{\pi}{2}\widehat\tau_{ij})$, which is a natural plug-in estimator of the correlation based on the main result of \cite{lindskog2003kendall}.
\end{rem}

To illustrate usefulness of this general approach we list a number of situation that are included in this set-up. 

\textbf{Graphical lasso and SCAD penalties:} If $|L_{ij}|=|U_{ij}|=\rho>0$ for all $i\neq j$, \eqref{eq:pospenalR} corresponds to the standard graphical lasso. More generally, if $|L_{ij}|=|U_{ij}|=\rho_{ij}$, that is when $\|K\|_{LU}$ is a norm, we obtain a version of the graphical lasso that takes into account different scalings of the variables. This general version is used in the adaptive GLASSO procedure and the local linear approximation algorithm used for general concave penalties rely on solving one or more problems of this form; see \cite{scad} for details.

\textbf{Asymmetric graphical lasso:} If $L_{ij}=-\rho_-$ and $U_{ij}=\rho_+$ for all $i\neq j$ where $0<\rho_-,\rho_+<+\infty$, we obtain a version of the graphical lasso, where positive entries of $K$ are penalized at a different rate than the  negative entries.

\textbf{Positive graphical lasso:} If $L=0$ and $U_{ij}=\rho$ for all $i\neq j$ then \eqref{eq:pospenalR} corresponds to the positive graphical lasso problem in \eqref{eq:pospenal}. This looks like the asymmetric graphical lasso problem with $\rho_-=0$ but as we will see, a zero penalty introduces additional complications concerning existence of the optimum. 

\textbf{$\mtp$ distributions:} If $L=0$ and $U_{ij}=+\infty$ for all $i\neq j$ then \eqref{eq:pospenalR} gives the maximum likelihood estimator for constrained M-matrix estimation. In Remark~\ref{rem:tofinite} we show that, rather than infinite, $U_{ij}$ must be sufficiently large.

\textbf{Gaussian graphical models:} In certain situations we may  in  advance wish to specify that some entries of $K$ should be zero. If $K_{ij}=0$, this can be imposed by setting $L_{ij}=-\infty$, $U_{ij}=+\infty$ (by the standard convention $\pm \infty \cdot 0=0$). Thus the optimization algorithm detailed in Section~\ref{sec:pglasso_alg} also yields an interesting alternative to the IPS algorithm and other edge based algorithms in \cite{speed1986gaussian} which may occasionally be slow.

\textbf{Dual likelihood estimate:}
The optimization problem in \eqref{eq:dualMLE} is equivalent to \eqref{eq:pospenalR} just replacing $K$ with $\Sigma$, $S$ with $\widehat{K}$, setting $L_{ij}=-\infty$ and $U_{ij}=0$. 
\begin{prop}\label{prop:kkt}
    If the optimum in \eqref{eq:pospenalR} exists, it is the unique positive definite matrix $\widehat K$ (with $\widehat \Sigma=\widehat K^{-1}$) satisfying 
    $$
\widehat\Sigma_{ij}-S_{ij}\;\in \;\begin{cases}
\{L_{ij}\} & \mbox{if } \widehat K_{ij}<0,\\
[L_{ij},U_{ij}] & \mbox{if } \widehat K_{ij}=0,\\
\{U_{ij}\} & \mbox{if } \widehat K_{ij}>0.
\end{cases}\qquad\mbox{for all } i,j\in V.
$$
\end{prop}
\begin{proof}
The subgradient of the function $\max\{L_{ij}K_{ij},U_{ij}K_{ij}\}$ at $K_{ij}=0$ is the interval $[L_{ij},U_{ij}]$. This subgradient is $\{L_{ij}\}$, $\{U_{ij}\}$ if $K_{ij}<0$, $K_{ij}>0$ respectively. Now the conclusion follows from the standard theory for non-differentiable convex functions; see, for example, \cite{rockafellar1970convex}. 
\end{proof}

The problem  \eqref{eq:pospenalR} is a convex optimization problem. Its dual problem is particularly  simple and admits a straight-forward block-coordinate descent procedure. To obtain the dual, note that $$\max\{L_{ij}K_{ij},U_{ij}K_{ij}\}=\sup_{L_{ij}\leq \Gamma_{ij}\leq U_{ij}}\Gamma_{ij}K_{ij}$$ and so
$$
\|K\|_{LU}\;=\;\sup_{L\leq \Gamma\leq U} \tr(\Gamma K)
$$
whereby \eqref{eq:pospenalR} becomes  
$$
\inf_{K\succ 0}\sup_{L\leq \Gamma\leq U} \big\{-\log\det K+\tr((S+\Gamma)K)\;\big\}.
$$
Swapping $\inf$ with $\sup$ and using the fact that the unique optimizer with respect to $K$ of $-\log\det K+\tr((S+\Gamma)K)$ (if exists) is $(S+\Gamma)^{-1}$, we obtain the dual problem by letting $\Sigma=S+\Gamma$:
\begin{equation}\label{eq:pospenalRdual}
{\rm maximize  }\;\;  \log\det \Sigma  +d\quad\mbox{subject to }S+L\leq \Sigma\leq S+U.
\end{equation}
In particular, every feasible point of the dual problem \eqref{eq:pospenalRdual} has the same diagonal as $S$\footnote{Our setting can be easily extended to the case when the diagonal entries of $K$ are also penalized. In this case the optimal point of the dual problem has the same diagonal as $S+U$.}. In particular, if the correlation matrix $R$ is used as the data input, the optimum is a correlation matrix too.

\begin{rem}\label{rem:tofinite}
Since $\Sigma$ is positive definite we have in particular that $-\sqrt{S_{ii}S_{jj}}<\Sigma_{ij}<\sqrt{S_{ii}S_{jj}}$ for all $i\neq j$. It follows that every dually feasible $\Sigma$ satisfies
$$
\max\{S_{ij}+L_{ij},-\sqrt{S_{ii}S_{jj}}\}\;\leq\;\Sigma_{ij}\;\leq\;\min\{S_{ij}+U_{ij},\sqrt{S_{ii}S_{jj}}\}.
$$
This allows us to replace $U_{ij}$ with $U_{ij}\wedge (\sqrt{S_{ii}S_{jj}}-S_{ij})$ and $L_{ij}$ with $L_{ij}\vee (-S_{ij}-\sqrt{S_{ii}S_{jj}})$, which is particularly useful if $L_{ij}=-\infty$ or $U_{ij}=+\infty$. 
\end{rem}

\begin{cor}\label{cor:ineqszeros}Let $\widehat K$ be the optimal solution to \eqref{eq:pospenalR}. If $L_{ij}\leq -S_{ij}-\sqrt{S_{ii}S_{jj}}$ then $\widehat{K}_{ij}\geq 0$. If $U_{ij}\geq -S_{ij}+\sqrt{S_{ii}S_{jj}}$ then $\widehat{K}_{ij}\leq 0$. In particular, $\widehat K_{ij}=0$ if both conditions hold.
\end{cor}

Since \eqref{eq:pospenalR} is always feasible, feasibility of the dual problem \eqref{eq:pospenalRdual} assures that the optimum of \eqref{eq:pospenalR} exists and is unique. We show below that it always holds if $L$ and $U$ have no zeros outside of the diagonal. Under minor conditions this also holds for the positive graphical lasso in which case $L$ is a zero matrix. We provide a more detailed treatment of this problem in Section~\ref{sec:start}. But first we introduce our optimization algorithm.

\subsection{Optimization}\label{sec:pglasso_alg}

To solve \eqref{eq:pospenalRdual} we use a straightforward block-coordinate descent approach that is a direct modification of the algorithm used for the dual graphical lasso problem in \cite{banerjee2008model}. An important difference is that, by default,  we do not penalize the diagonal, which leads to additional issues that may  arise. We optimize the determinant of $\Sigma$ updating row by row, but keep the diagonal entries fixed to be equal to the diagonal of $S$.

For the $j$-th row we consider $\log\det\Sigma$ as the function of $\Sigma_{j,\- j}$ keeping  the other entries of $\Sigma$ fixed. By standard matrix algebra
$$\log\det\Sigma\;=\;\log\det\Sigma_{\- j}+\log\det(\Sigma_{jj}-\Sigma_{j,\- j}(\Sigma_{\- j})^{-1}\Sigma_{\- j,j}).$$
Thus maximizing $\log\det\Sigma$ with respect to $y:=\Sigma_{j,\- j}$ is equivalent to minimizing $y^T (\Sigma_{\- j})^{-1} y$, where we need to impose linear conditions that $S_{ij}+L_{ij}\leq y_i\leq S_{ij}+U_{ij}$ for every $i\in V\setminus \{j\}$. This is an instance of a quadratic program. The following result is a straight-forward generalization of \cite[Theorem 4]{banerjee2008model}. It allows to quickly identify disconnected nodes in the underlying graph.
\begin{lem}\label{lem:isolated}
    If $S_{j,\- j}+L_{j,\- j}\leq 0\leq S_{j,\- j}+U_{j,\- j}$ for some $j\in V$  then $\widehat{\Sigma}_{j,\- j}=\widehat K_{j,\- j}=0$.
\end{lem}

The starting point $\Sigma^0$ of the algorithm needs to be chosen carefully so  that $\Sigma^0$ is dually feasible. In this case, each iterate of the algorithm is guaranteed to be dually feasible. In Section~\ref{sec:start} we show how to find such a starting point. Given the starting point, our procedure is straightforward and described in Algorithm~\ref{alg:EPlasso}. Corollary~\ref{cor:ineqszeros} and Lemma~\ref{lem:isolated} give an obvious way to speed up the computations but reducing the number of nodes that have to be visited at each step. To solve the quadratic problem in each iteration we use the  \texttt{quadprog} package in \texttt{R}.

\bigskip

\begin{algorithm}[htb]
  \SetAlgoLined
  \KwData{A positive semidefinite matrix $S$, penalty matrices $L\leq 0\leq U$.}
  \KwResult{A maximizer of \eqref{eq:pospenal}.}
  Initialize: $\Sigma=\Sigma^0$ (a dually feasible point)\;
  \While{there is no convergence}{
  \For{$j=1,\ldots,d$}{
  Update $\Sigma_{j,\setminus  j} \leftarrow \widehat y$, where
  $$\widehat y =\arg\min_y\Big\{y^T(\Sigma_{\setminus  j})^{-1}y: \; S_{j,\setminus j}+L_{j,\setminus j}\leq y\leq S_{j,\setminus j}+U_{j,\setminus j}\Big\}.$$
      }}
  \caption{The Graphical Oriented LAZy Optimization (GOLAZO) Algorithm.}\label{alg:EPlasso}
\end{algorithm}	
\bigskip
\noindent To establish convergence in Algorithm~\ref{alg:EPlasso} we track the duality  gap
$$
\tr(SK)-d+\|K\|_{LU},
$$
which is guaranteed to be non-negative for each step of the algorithm, decrease at each iteration, and to be zero at the optimum. We stop the algorithm once this positive gap becomes sufficiently close to zero. 

In the actual implementation it is important to  compute the dual gap carefully in case $L,U$ contain infinite entries. We simply  make use of Remark~\ref{rem:tofinite} and replace $\pm \infty$ with appropriate finite bounds. The only remaining issue to complete the description of the algorithm is a procedure to obtain a dually feasible starting point. We address this issue in the next section.

\subsection{Dually feasible starting point}\label{sec:start}

Recall that $\Sigma$ is dually feasible if $\Sigma$ is positive definite and $L\leq \Sigma-S\leq U$. If $S$ is positive definite, it is dually feasible, and we take $\Sigma^0=S$. We then focus on the case when $S$ is rank deficient. Since $S$ is a sample covariance matrix it is always positive semidefinite and it has rank at least one. This implies that (with probability one) the diagonal entries  are strictly positive. Clearly, if some $S_{ii}=0$ then no feasible point exists. We then always assume that $S_{ii}>0$ for all $i\in V$.

We first show how to construct a starting point in the  case when for each $i\neq j$ both the negative and the positive values of $K_{ij}$ are penalized. Denote by ${\rm diag}(S)$ the diagonal matrix whose diagonal is equal to the diagonal of $S$.
\begin{lem}\label{lem:startdiagonal}
If $ L_{ij}<0<U_{ij}$ for all $i\neq j$  then there exists $t\in (0,1)$ such that  $\Sigma^0=(1-t)S+t\,{\rm diag}(S)$ is dually feasible. The condition for dual feasibility becomes  that $L_{ij}\leq -t S_{ij}\leq U_{ij}$ for all $i\neq j$. 
\end{lem}
\begin{proof}
Since ${\rm diag}(S)$ is positive definite, $\Sigma^0=(1-t)S+t\,{\rm diag}(S)$ is positive definite for every $t\in (0,1]$. We have $\Sigma^0-S=t\,({\rm diag}(S)-S)$ and so $(\Sigma^0-S)_{ij}=-tS_{ij}$ for all $i\neq j$. Since $L_{ij}<0$ and $U_{ij}>0$, it holds that if $t$ is sufficiently small then $L_{ij}\leq t\,({\rm diag}(S)-S)_{ij}\leq U_{ij}$ for all $i\neq j$.
\end{proof}

The conditions of Lemma~\ref{lem:startdiagonal} are satisfied for the graphical lasso and this result can be turned into an explicit procedure for computing a starting point. Note that using \texttt{glasso} or \texttt{EBICglasso} in \textsc{R} outputs either a warning or an error when $S$ is not positive definite. As a consequence of Lemma~\ref{lem:startdiagonal}, we have: 
\begin{thm}\label{th:lassoexists}
The optimum in the  graphical lasso problem always exists and is unique. More generally, this also holds for the GOLAZO if $L_{ij}<0<U_{ij}$ for all $i\neq j$.
\end{thm}

The conditions of Lemma~\ref{lem:startdiagonal} are not satisfied for the positive graphical lasso and, in general, for any case where $S$ lies on the boundary of the rectangle $\{\Sigma:L\leq \Sigma-S\leq U\}$. If $L$ has potentially zero entries but $U_{ij}>0$ for all $i\neq j$ then we still can prove that a dually feasible point exists under very mild additional assumptions on $S$. 
This relies on an explicit construction of a dually feasible point based on the definition of a single-linkage matrix of $S$ given in \cite{LUZ}. We provide  the  details of this construction in Appendix~\ref{app:Zmatrix}, which also contains the proof of the following result.
\begin{thm}\label{th:plassoexists}
The optimum in the  positive graphical lasso problem exists and is unique as long as $S_{ij}<\sqrt{S_{ii}S_{jj}}$ for all $i\neq j$, which happens with probability  one if the sample size is at  least two. The same holds in general whenever $U_{ij}>0$ for all $i\neq j$.
\end{thm}

\subsection{Choice of the penalty parameter}\label{sec:ebic}

In this section we propose a simple method to choose the penalty parameters in $L$ and $U$. Our method is based on the Extended Bayesian Information Criterion (EBIC) proposed in \cite{chen2008extended} and adapted to  problems of graphical lasso type in \cite{foygel2010extended}.  Given a sample of $n$ independent and identically distributed observations, let $\widehat  K$ denote an estimate and $\widehat E$ the set of edges of  the underlying graph of $\widehat{K}$. The  EBIC criterion takes the form
$$
\bic_\gamma(\widehat  E)\;=\;-n\ell_n(\widehat{K})+|\widehat E|(\log(n)+4\gamma\log(d)).
$$
The criterion is indexed by a parameter  $\gamma\in [0,1]$. If $\gamma=0$ then the classical BIC is recovered, which is known to be asymptotically consistent for model selection in case $d$ is fixed and $n$ goes to infinity. Positive values of $\gamma$ lead to better graph estimates in  the case when $d$ and $n$ are comparable. This observation can be formalized in certain scenarios but otherwise relies on numerical experiments; cf.\ \cite{foygel2010extended}.

The model selection procedure based on EBIC  relies on choosing a sequence of potential penalty  parameters $\rho_1,\ldots,\rho_l$. Then for fixed $L,U$ we then compute $l$ optima for the GOLAZO problem with parameters $\rho_i L, \rho_i U$. For each of these we compute EBIC and choose $\rho_i$ that minimizes this criterion. For positive glasso the canonical choice is $L_{ij}=0$ and $U_{ij}=1$ for all $i\neq j$. Finally note that this procedure is trivially parallelizable, which we exploit in our code.

\subsection{Sign-constrained likelihood optimization}\label{sec:m-matrix}

The positive graphical lasso problem links to the problem of maximum likelihood estimation under M-matrix constraints \cite{slawski2015estimation,LUZ}, that is to the problem
$$
{\rm minimize}\;\;-\ell(K)\qquad\mbox{subject to } K_{ij}\leq 0\mbox{ for all }i\neq j.
$$
More generally, if $L_{ij}\in \{-\infty,0\}$ and $U_{ij}\in \{0,+\infty\}$ then \eqref{eq:pospenalR} is equivalent to optimizing the Gaussian likelihood under sign-constraints. If $E_L$ is the set  of all $i\neq j$ such that $L_{ij}=-\infty$ and $E_U$ is the set  of all $i\neq j$ such that $U_{ij}=+\infty$ then \eqref{eq:pospenalR} amounts to maximizing the Gaussian likelihood over the set of all inverse covariance matrices $K$ such that $K_{ij}\geq 0$ for all $ij\in E_L$ and $K_{ij}\leq 0$ for all $ij\in E_U$. If $ij\in E_L\cap E_U$ then $K_{ij}=0$.

One of the important reasons why we choose not to penalize the diagonal of $K$ is the following result.

\begin{thm}\label{th:InfgivesMLE}
    Let $L,U$ be such that $L_{ij}\in \{-\infty,0\}$ and $U_{ij}\in \{0,+\infty\}$ for all $i\neq j$. The GOLAZO estimator is equal to the  maximum likelihood estimator under the constraints $K_{ij}\geq 0$ for all $ij\in E_L$ and $K_{ij}\leq 0$ for all $ij\in E_U$.
\end{thm}
Recall from Remark~\ref{rem:tofinite}  that infinity can be replaced by a sufficiently large number.

\subsection{A link to graphical models}

Let $\widehat G$ denote the graph determined by the support of $\widehat K$. Proposition~\ref{prop:kkt} implies the following result.
\begin{cor}\label{cor:optimum}
The optimum $\widehat K$ of \eqref{eq:pospenalR} is equal to the maximum likelihood estimator in the Gaussian graphical model $M(\widehat G)$ with the sufficient statistics $(\widehat S_{ij})_{ij \in \widehat G}$, where $\widehat S_{ij}=S_{ij}+L_{ij}$ if $\widehat K_{ij}<0$ and $\widehat S_{ij}=S_{ij}+U_{ij}$ if $\widehat K_{ij}>0$.
\end{cor}
\begin{proof}
We have that $\widehat\Sigma\in M(\widehat G)$ and it coincides with $\widehat S$ on the edges of $\widehat G$. It is then the unique optimum of the MLE problem.
\end{proof}
This result implies that as soon as we have access to the matrix of signs of $\widehat K$ with ${\rm sgn}(\widehat K_{ij})\in \{-1,0,1\}$ we could alternatively find $\widehat K$ by augmenting the sufficient statistics and fitting the corresponding Gaussian graphical model. If $\widehat G$ is decomposable, the optimum is then given in a closed form. 
\color{black}

\section{Applications and simulations}\label{sec:appls}

In this section we illustrate our methods with  two applications and a small simulation study. The computations are based on the GOLAZO algorithm described in Section~\ref{app:golazo} as implemented in the R package \texttt{golazo} available on GitHub. The corresponding R Markdown files can be downloaded from \url{http://econ.upf.edu/~piotr/supps/2020-LZ-golazo.zip}. 

\subsection{Body fat data}

As a simple illustration for our method we analyse the body fat data first studied in \cite{johnson1996fitting} and available in the \texttt{R} package  \texttt{gRim}.
\begin{figure}[htb]
    \centering
    \includegraphics[scale=.68]{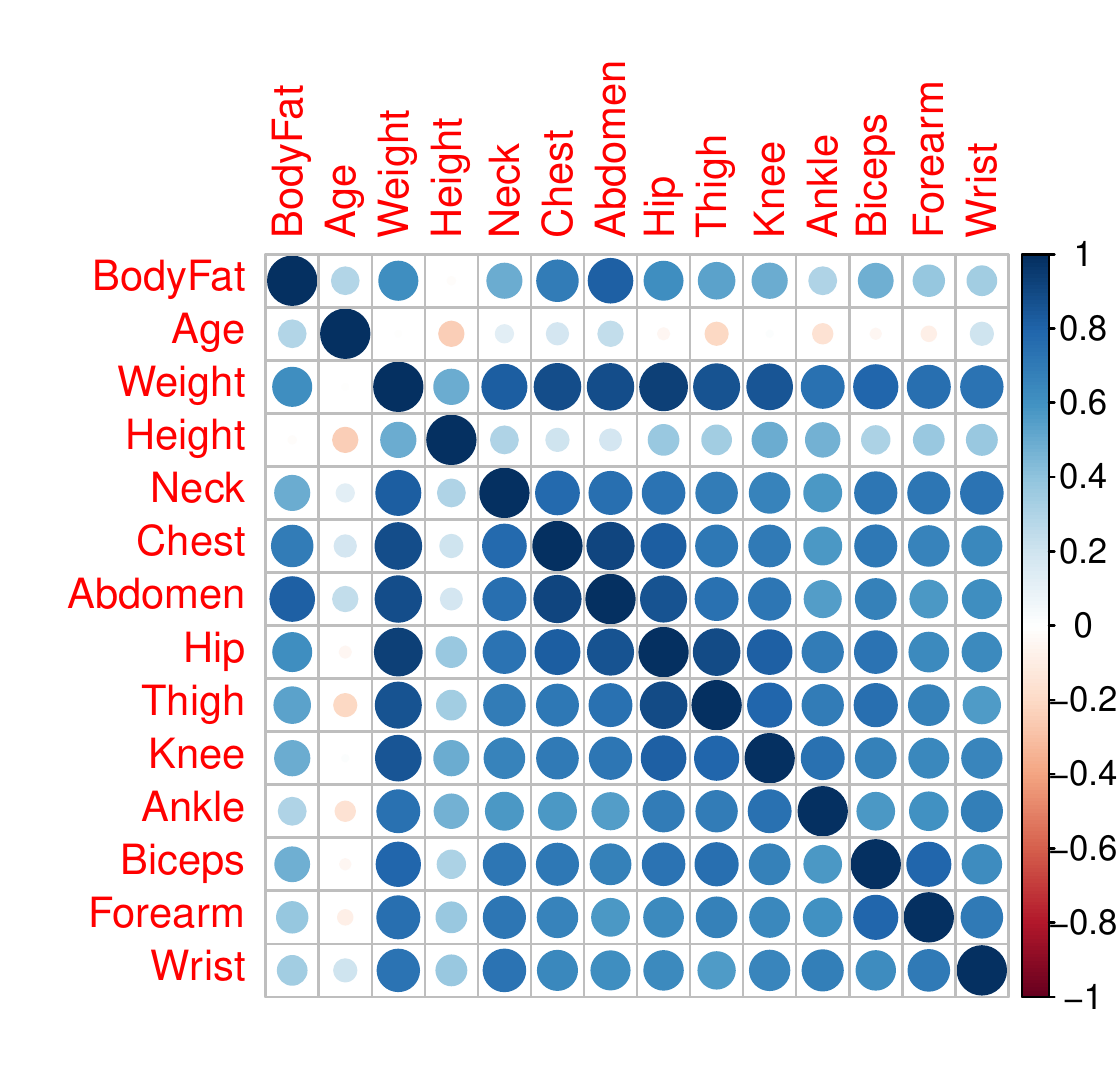} \includegraphics[scale=.68]{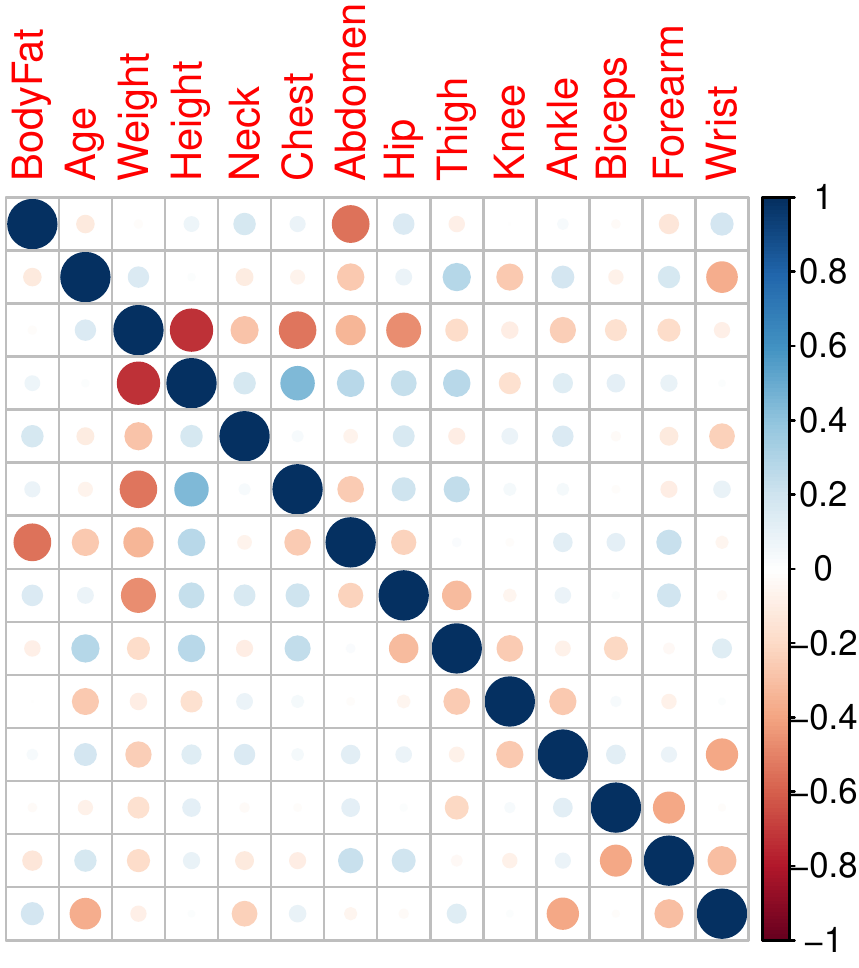}
    \caption{The sample correlation matrix for the body fat data (left) and its normalized inverse (right).}
    \label{fig:RBodyFat}
\end{figure}

In the data, percentage of body fat, age, weight, height, and ten body circumference measurements are recorded for 252 men. We remove 11 individuals from the  study for various problems reported for their respective  observations. As shown in Figure~\ref{fig:RBodyFat} the variables in this dataset are  strongly positively correlated. The only  exception is Age. Its negative correlation with height reflects the fact that the growth in wealth, especially after the second world war, has made the new generations taller. Also, the negative correlation with thigh and ankle could reflect that muscle mass tends to be reduced with age. {In other words, the relation between Age and the other variables is complex and certainly not approximately linear.} For these reasons and for simplicity, we remove Age from our analysis. 

\begin{figure}[htb]
    \centering
    \includegraphics[scale=.6]{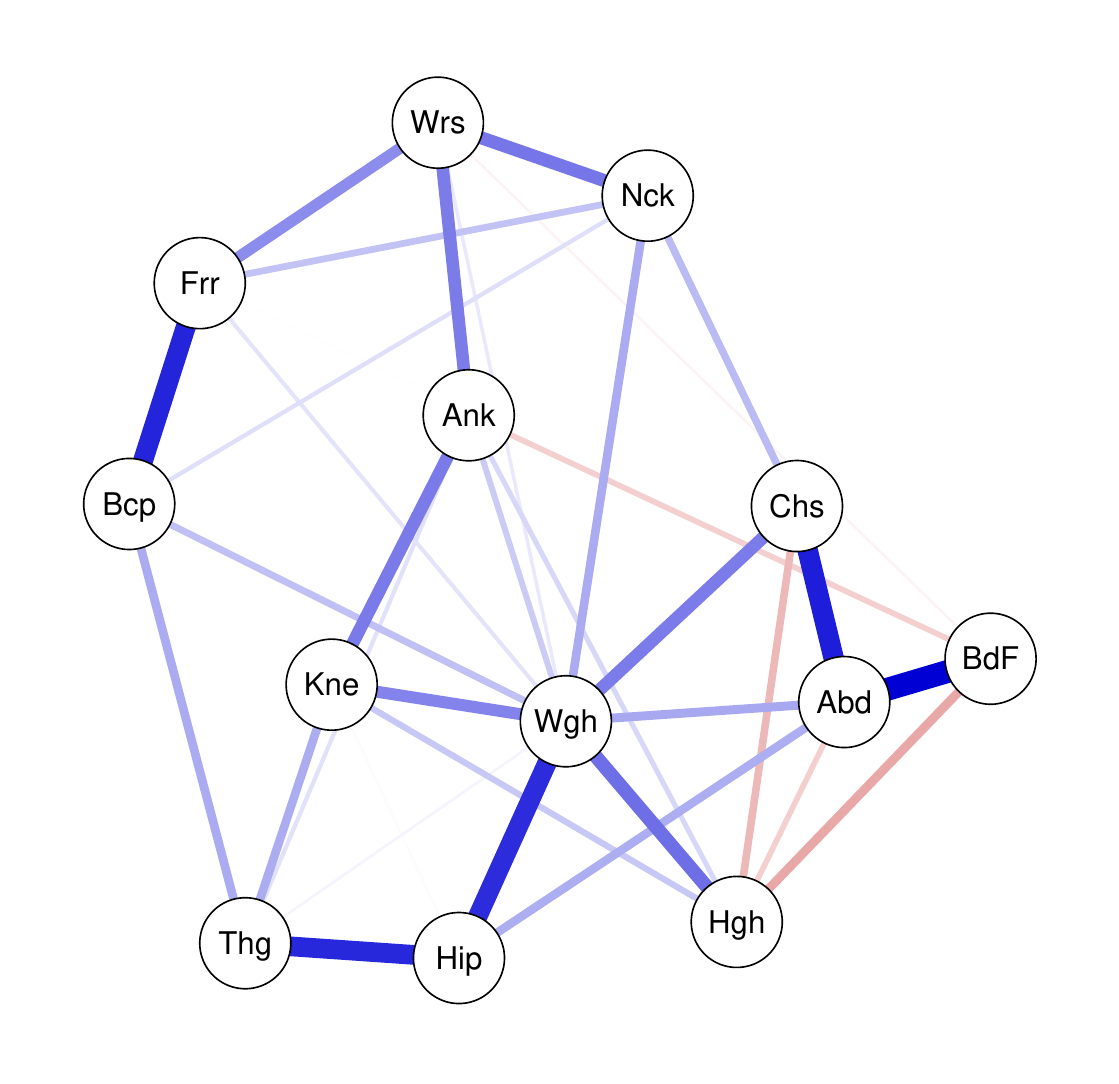}\includegraphics[scale=.6]{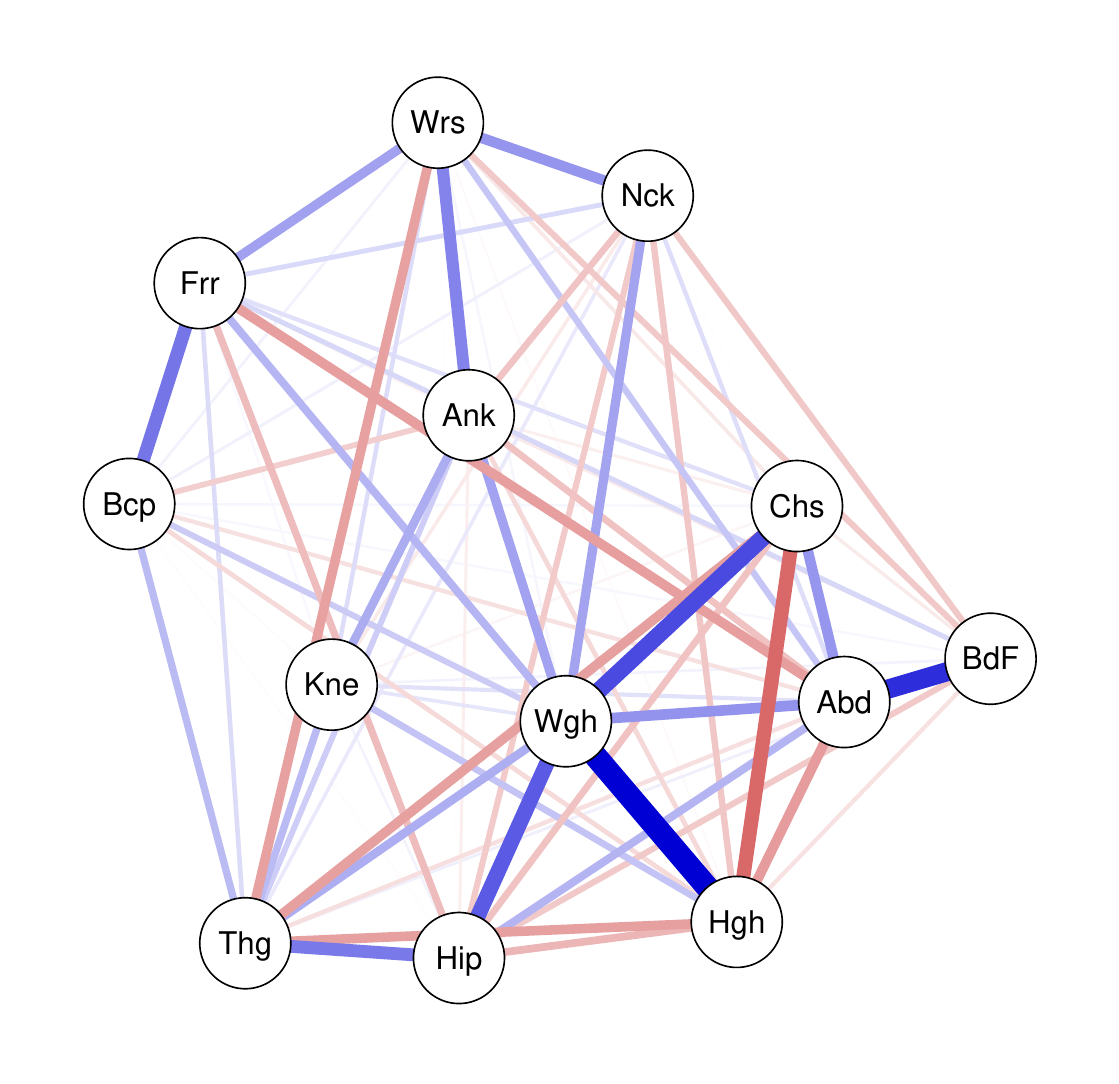}
    \caption{Partial correlations in the estimated positive glasso graph (left) and glasso graph (right).}
    \label{fig:BodyFatK}
\end{figure}
As the inverse of the sample covariance matrix has a significant number of positive entries, $\mtp$ seems to be too strong hypothesis for this dataset. After  normalizing the data we run  the positive graphical lasso procedure with $\rho=0.11$. This choice was based on the  EBIC criterion with parameter $\gamma=0.5$ as described in Section~\ref{sec:ebic}. Figure~\ref{fig:BodyFatK} shows the graph (together with signs of the partial correlations) of the optimum $\widehat{K}$ (left) compared with the graphical lasso estimate. The penalty  in the  graphical lasso estimate  is chosen close to zero both by cross-validation and by the EBIC criterion. Although the  positive glasso estimate gives lower likelihood than  the glasso estimate, it is much sparser and beats glasso in the EBIC criterion with $\gamma=0.5$: $-364.9$ for the positive glasso and $-237.65$ for the graphical lasso. 

In the second step of our procedure we take the  resulting estimate $\widehat  K$ of the positive glasso procedure and compute the dual MLE $\widecheck{\Sigma}$ under edge positivity to further regularize the positive glasso estimate. However, although $\widehat{K}$ is not an M-matrix, it corresponds to a locally associated distribution and so the second step of  the algorithm becomes redundant.  

\subsection{Positive co-expression gene network}\label{sec:gene}

As we argued in the  introduction, our main motivation was to develop statistically sound methods for building gene expression networks that focus on positive co-expression. For illustrative purposes we focus on a relatively small subsample of genes. We analyse a publicly available microarray expression data profiling umbilical cord tissue; cf.\ \cite{cohen2007perturbation,costa2016umbilical}. From  \url{https://functionalgenomics.upf.edu/supplements/FIRinELGANs} we downloaded the normalized and filtered gene expression data, as well as its corresponding phenotype data, including a batch indicator variable that specifies the groups in which samples were processed, birth weight, gestational age, sex and fetal inflammatory response (FIR) status. 

As described in the introduction, we focus  our analysis on 136 genes 
which were coordinately upregulated in FIR-affected infants to trigger an innate immune response, and therefore, we can assume their positive co-expression.

We first run the GOLAZO procedure penalizing negative partial correlations with an EBIC optimal penalty parameter. This computation takes less than a minute on  a standard laptop. The penalty $\rho=0.5$ was chosen as optimal with respect to the EBIC criterion with $\gamma=0.5$. The resulting estimate $\widehat{K}^\rho$ is very sparse with  the edge density $0.067$ (an entry of $\widehat{K}^\rho$ is always treated nonzero if its absolute value  exceeded $10^{-6}$).  Still the diameter of this graph, displayed in Fig.~\ref{fig:genegreen},  is very small, just  5.
\begin{figure}[htb]
\centering\includegraphics[width=0.8\textwidth]{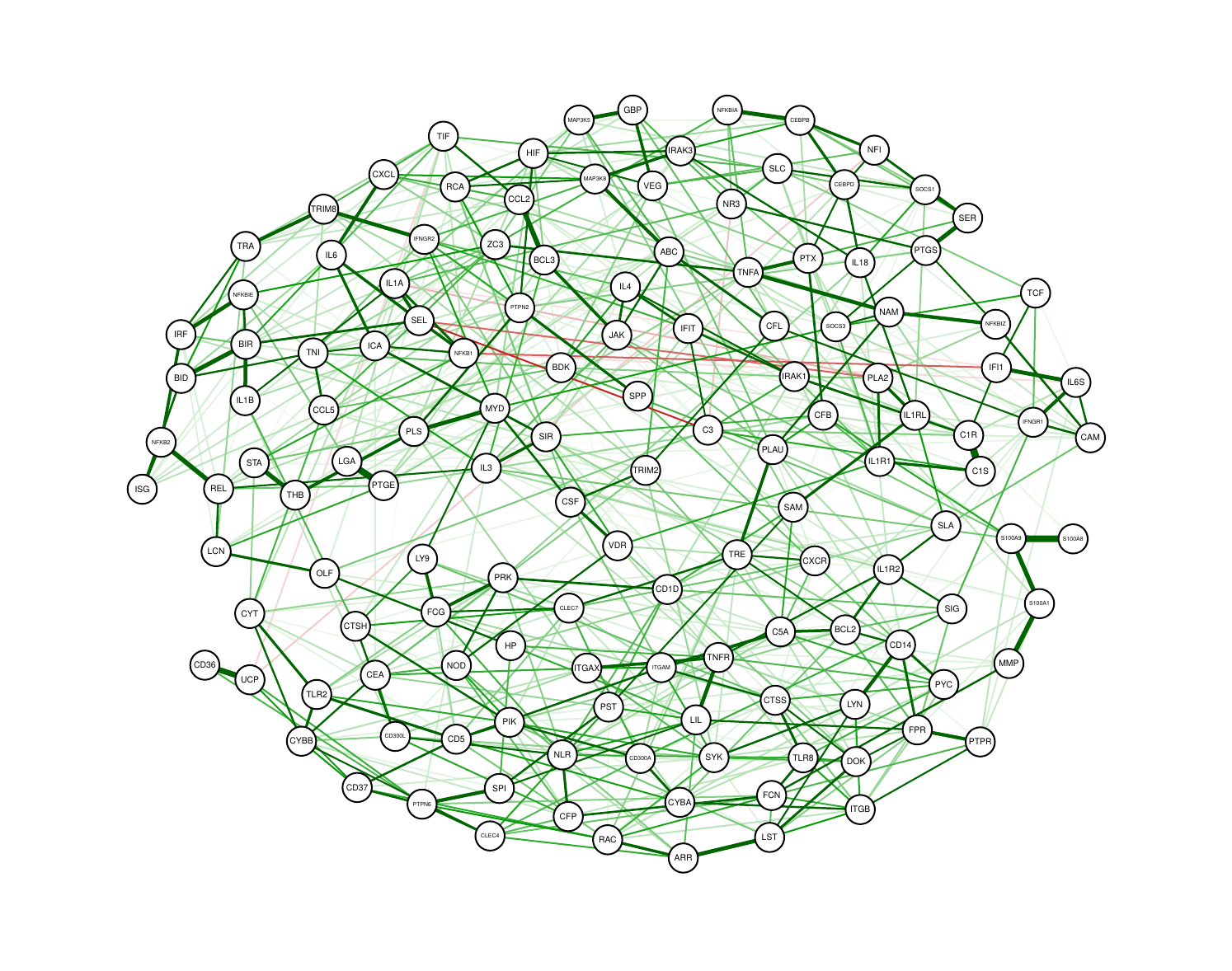}
\caption{\label{fig:genegreen} The concentration graph for the co-expression gene network. Positive partial correlations are indicated with green color and negative partial correlations with red color. The thickness of edges is proportional to their absolute size.}
\end{figure}

Like in the Body Fat example in the previous section, also here the  EBIC criterion chooses much smaller penalty parameter for the standard graphical lasso, $\rho=0.1$. The EBIC criterion for the  optimal positive graphical lasso model penalty is $-8784$ on the other hand, the graphical lasso gives a much denser graph with edge density $0.21$ and the EBIC is also much  higher, $13478$. The optimum $\widehat{K}^\rho$ is not an M-matrix and there are several significant negative partial correlations, displayed in red in Fig.~\ref{fig:genegreen}. However, the distribution is already locally associated so the second step of our procedure is again redundant. This also confirms that local association is a reasonable assumption for this dataset.

\begin{appendix}
\section{Proof of Theorem~\ref{th:equal}}\label{app:proof}
\subsection{Chernoff regularity and convexity}

The asymptotic analysis of statistical procedures under constraints typically involves technical assumptions on the local geometry of the constrained space around the true parameter $\bs \psi_0$. Conditions of this form are called Chernoff regularity conditions; cf.\ \cite{geyer1994asymptotics} and references therein. In our case, convexity ensures these conditions to hold, but we provide the relevant definitions for completeness. 

\begin{defn}
The \emph{tangent cone} $T_C(\bs \psi_0)$ of the set $C\subseteq \R^k$ at the point $\bs \psi_0$ is the set of vectors in $\R^k$ that are limits of sequences $\alpha_n(\bs\psi_n-\bs \psi_0)$, where $\alpha_n$ are positive reals and $\bs\psi_n\in C$ converge to $\bs \mu_0$.
\end{defn}
\begin{defn}
The set $C\subseteq \R^k$ is \emph{Chernoff regular} at $\bs\psi_0$ if for every vector $\bs\tau $ in the tangent cone $T_C(\bs\psi_0)$ there exists $\varepsilon>0$ and a map $\alpha:[0,\varepsilon) \to C$ with $\alpha(0)=\bs\psi_0$ such that $\bs \tau=\lim_{t\to 0^+}[\alpha(t)-\alpha(0)]/t$. In this case we say that $T_C(\bs\psi_0)$ is \emph{derivable}; cf.\ Definition 6.1 in \cite{rockafellar2009variational}.
\end{defn}
The standard asymptotic results typically assume Chernoff regularity. We will use the following result.
\begin{thm}[Theorem~6.9, \cite{rockafellar2009variational}]\label{th:RWchernoff}
A convex set $C\subseteq \R^k$ is Chernoff regular at any $\bs\psi_0\in C$. 
\end{thm}
It is clear from the definition that Chernoff regularity is preserved under a smooth and regular change of variables $G:\R^k\to \R^k$. Indeed,  the tangent cone $T_{G({C})}(G(\bs \psi_0))$  is equal to $\nabla G(\bs\psi_0)\cdot T_C(\bs \psi_0)$; cf.\ Section~6.C in \cite{rockafellar2009variational}. If $\bs\tau\in T_{G({C})}(G(\bs \psi_0))$ and $T_C(G(\bs \psi_0))$ is derivable then $$(\nabla G(\bs\psi_0))^{-1}\cdot \bs\tau\;=\;\lim_{t\to 0^+}[\alpha(t)-\alpha(0)]/t$$
for some $\alpha:[0,\epsilon)\to C$. Then $\alpha_G=(\nabla G(\bs\psi_0))\cdot \alpha:[0,\epsilon)\to G(C)$ can be used to show that $\bs\tau$ is derivable.

\subsection{Asymptotics of the maximum likelihood and mixed dual estimator}
Recall our notation $\bs t_n=\sum_{i=1}^n \bs t(X^{(i)})/n$. It is a standard result that $\sqrt{n}(\bs t_n-\bs\mu_0)$ is asymptotically normally distributed; see Proposition~4.3 in \cite{sundberg}. In this section we show that the maximum likelihood  estimator under a convex restriction has a similar rate of convergence, with the limiting distribution not necessarily being  normal. Using equivariance of the MLE and the delta method, we can show this {also holds for the MLE and MDE in a mixed convex exponential family. } 

Let $\bs \psi=\psi(\bs\theta)$ be an alternative smooth, regular, and bijective parametrization of the exponential family $\cE$, so that, $\bs\theta=\psi^{-1}(\bs\psi)$. The log-likelihood function expressed in this new parametrization is denoted by $\ell(\bs \psi;\bs t_n)=\<\psi^{-1}(\bs \psi),\bs t_n\>-A(\psi^{-1}(\bs \psi))$. 
\begin{prop}\label{prop:mleconvex}
   Consider an alternative smooth and regular parametrization $\bs \psi=\psi(\bs\theta)$ of the exponential family $\cE$. Let $C$ be a closed and convex subset of the parameter space $\psi(\Theta)$. Let $\widehat{\bs\psi}_n=\arg\max_{\bs\psi\in C}\ell(\bs\psi;\bs t_n)$ be the maximum likelihood estimator over $C$. If the data are generated from the distribution with parameter $\bs\psi_0\in C$ then $\sqrt{n}(\widehat{\bs\psi}_n-\bs\psi_0)$ converges in distribution.
\end{prop}
\begin{proof}
Since the maximum likelihood estimator is equivariant, we have that $\widehat{\bs \psi}_n=\psi(\widehat{\bs\theta}_n)$, where $\widehat{\bs \theta}_n=\arg\max_{\bs\theta\in \psi^{-1}(C)}\ell(\bs\theta;\bs t_n)$. By the delta method (see Theorem~3.1 in \cite{vandervaart}), it is enough to show that $\sqrt{n}(\widehat{\bs \theta}_n-\bs\theta_0)$ converges in distribution. Since the MLE is an M-estimator, this follows from Theorem~4.4 in \cite{geyer1994asymptotics}. This theorem uses a number of assumptions that we verify one by one: Assumption A   holds because the function $F(\bs \theta)=A(\bs \theta)-\<\bs\theta,\bs\mu_0\>$ admits a quadratic approximation around $\bs\theta_0$ with a positive definite Hessian $\nabla^2 A(\bs \theta_0)$ and $\nabla F(\bs\theta_0)=0$.  Assumption~B is satisfied simply because the second derivative of the likelihood does not depend on the data at all and so condition (4.3) in \cite{geyer1994asymptotics} trivially holds. Assumption~C requires that the standard central limit theorem for the gradient of the likelihood function holds, which again is automatic for exponential families. Assumption~D holds simply because $\widehat{\bs \theta}_n$ is the exact minimizer of $F_n(\bs\theta)=-\ell(\bs\theta;\bs t_n)$ over $\psi^{-1}(C)$. Finally, Chernoff regularity of $\psi^{-1}(C)$ at $\bs\theta_0$ follows by convexity of $C$ and the fact that the property is invariant under smooth transformations; cf.\ Theorem~\ref{th:RWchernoff} and the discussion below it.
\end{proof}

As in Section~\ref{sec:asym}, we now reserve the notation $\bs\psi$ for the underlying mixed parametrization, $\bs\psi=(\bs\mu_u,\bs\theta_v)$. The maximum likelihood estimator in our problem is obtained by maximizing $\ell(\bs \psi;\bs t_n)$ over all parameters $\bs\psi$ in the mixed convex exponential family $\cE'$, $\bs\psi\in M_u'\times \Theta_v'$. This MLE is denoted $\widetilde{\bs \psi}_n$ ($\widetilde{\bs \theta}_n$, $\widetilde{\bs \mu}_n$ resp.) to distinguish from the estimator $\widehat{\bs \theta}_n$ obtained in step (S1) of the procedure for finding the MDE. Recall that  $\overline{\bs\psi}_{n}=(\bs u_n,\theta_v(\bs t_n))$ and denote $\widehat{\bs\psi}_n=(\bs u_n,\widehat{\bs \theta}_v)$, i.e.\ $\widehat{\bs\psi}_n$ is  $\widehat{\bs\theta}_n=(\widehat{\bs\theta}_u,\widehat{\bs\theta}_v)$ expressed in the mixed parametrization, where the fact that the first component of $\widehat{\bs\psi}_n$ is equal to $\bs u_n$ follows from Theorem~\ref{th:identity}.

\begin{cor}\label{cor:indistr}
The sequences $\sqrt{n}(\overline{\bs\psi}_{n}-\bs\psi_0)$, $\sqrt{n}(\widehat{\bs \psi}_n-\bs\psi_0)$, and $\sqrt{n}(\widetilde{\bs\psi}_n-\bs\psi_0)$ all converge in distribution. 
\end{cor}
\begin{proof}The estimators $\overline{\bs\psi}_n$, $\widehat{\bs\psi}_n$, and $\widetilde{\bs\psi}_n$ are all maximum likelihood estimators in families that satisfy the conditions in Proposition~\ref{prop:mleconvex}.
\end{proof}

The proof of our main result relies on the fact that the log-likelihood and dual log-likelihood have locally a similar shape around their global maximum. Moreover, the Hessian at this point is block diagonal. 

\begin{lem}\label{lem:infor}
Let $\bs t=(\bs u,\bs v)\in M$, $\overline{\bs\psi}=(\bs u,\theta_v(\bs t))$, and $\overline{\bs \theta}=\theta(\bs t)$. Then $\nabla_{\bs\psi} \ell(\overline{\bs\psi};\bs t)=\nabla_{\bs\psi} \widecheck{\ell}(\overline{\bs\psi};\bs t)=\bs 0$ and
$$
\nabla^2_{\bs\psi}\ell(\overline{\bs\psi};\bs t)\;=\;\nabla^2_{\bs\psi}\widecheck{\ell}(\overline{\bs\psi};\bs t)\;=\;-\begin{bmatrix}
{\rm var}(\bs u)^{-1} & \bs 0\\
\bs 0 & ({\rm var}(\bs t)^{vv})^{-1}
\end{bmatrix},
$$
where ${\rm var}(\bs t)^{vv}$ stands for the $\bs v\bs v$-block of ${\rm var}(\bs t)^{-1}$ and the variance is computed with respect to the distribution $P_{\overline{\bs \theta}}$.
\end{lem}
\begin{proof}
Slightly abusing notation, we write $\mu(\bs\psi)$ ($\theta(\bs\psi)$) for the map that maps the mixed parameter $\bs \psi$ to the corresponding mean (canonical) parameter. We have
$$
\nabla_{\bs \psi}\ell(\psi,\bs t)=\left(\frac{\partial \bs\theta}{\partial\bs\psi}\right)^T\!\!\cdot (\bs t-\mu(\bs\psi))
$$
$$
\nabla_{\bs \psi}\widecheck{\ell}(\psi,\bs t)=\left(\frac{\partial \bs\mu}{\partial\bs\psi}\right)^T\!\!\cdot (\overline{\bs \theta}-\theta(\bs\psi))
$$
from which the statement about the gradients easily follows by plugging $\bs\psi=\overline{\bs \psi}$ as $\mu(\overline{\bs \psi})=\bs t$ and $\theta(\overline{\bs\psi})=\overline{\bs \theta}$. The particular block-diagonal form of the Hessian   $\nabla^2_{\bs\psi}\ell(\overline{\bs\psi};\bs t)$ follows by Proposition~3.20 in \cite{sundberg}. It remains to show that $\nabla^2_{\bs\psi}\ell(\overline{\bs\psi};\bs t)\;=\;\nabla^2_{\bs\psi}\widecheck{\ell}(\overline{\bs\psi};\bs t)$. Using the Leibniz rule, we observe that differentiating $\nabla_{\bs \psi}\ell(\psi,\bs t)$ with respect to $\bs\psi$, we get one term that vanishes at $\bs\psi=\overline{\bs \psi}$ and so
$$
\nabla^2_{\bs\psi}\ell(\overline{\bs\psi};\bs t)\;=\;-\left(\frac{\partial \bs\theta}{\partial\bs\psi}\right)^T\frac{\partial \bs\mu}{\partial\bs\psi}.
$$
Using the same argument for $\nabla_{\bs \psi}\widecheck{\ell}(\psi,\bs t)$ we get 
$$
\nabla^2_{\bs\psi}\widecheck{\ell}(\overline{\bs\psi};\bs t)\;=\;-\left(\frac{\partial \bs\mu}{\partial\bs\psi}\right)^T\frac{\partial \bs\theta}{\partial\bs\psi}.
$$
Since both matrices are symmetric, the equality follows. 
\end{proof}

Using Lemma~\ref{lem:infor} with $\bs t=\bs t_n$, we see that the observed information  satisfies 
$$J_n:=-\nabla^2_{\bs\psi}\ell(\overline{\bs \psi}_n;\bs t_n)=-\nabla^2_{\bs\psi}\widecheck\ell(\overline{\bs \psi}_n;\bs t_n).$$
In particular, $J_n$ is always positive definite. Taking $\bs t=\widehat{\bs \mu}_n$ we get 
$$
\widehat J_n:=-\nabla^2_{\bs\psi}\ell(\widehat{\bs \psi}_n;\widehat{\bs \mu}_n)=-\nabla^2_{\bs\psi}\widecheck\ell(\widehat{\bs \psi}_n;\widehat{\bs \mu}_n).
$$
Both of $J_n$ and $\widehat J_n$ converge in probability to the Fisher information matrix since the mapping $\bs t\mapsto -\nabla^2_{\bs \psi}\ell(\overline{\bs\psi};\bs t)$ is continuous and both of $\bs t_n$ and $\widehat{\bs\mu}_n$ converge in probability to $\bs\mu_0$. %\piotr{Do we need a comment here?} 
We thus conclude
\begin{equation}\label{eq:Jhat}
    \widehat J_n\;=\;J_n+o_P(1).
\end{equation}
Denote by $\|\cdot\|_{J_n}$ ($\|\cdot\|_{\widehat J_n}$) the norm induced by the matrix $J_n$  ($\widehat J_n$ respectively), that is,
$$\|\bs x\|_{J_n}:=\sqrt{\bs x^TJ_n\bs x},\qquad \|\bs x\|_{\widehat J_n}:=\sqrt{\bs x^T\widehat J_n\bs x}.$$
Note that for every $\bs x\in \R^d$
$$
\lambda_{\min} (J_n) \|\bs x\|^2\;\leq\;\|\bs x\|^2_{J_n}\;\leq\; \lambda_{max} (J_n) \|\bs x\|^2,
$$
where $\lambda_{\min} (J_n)$ and $\lambda_{max} (J_n)$ are the minimal and the maximal eigenvalue of $J_n$. Since the eigenvalues are continuous functions of the matrix (use \cite[Theorem 2.6.4]{hornjoh}) and the fact that for a positive definite matrix singular values are equal to the eigenvalues), the sequences $\lambda_{\min}(J_n)$ and $\lambda_{\max}(J_n)$ converge in probability to the corresponding eigenvalues of the Fisher information matrix, which we denote by $\lambda^*_{\min},\lambda^*_{\max}$. The same argument applies to $\lambda_{\min}(\widehat J_n)$ and $\lambda_{\max}(\widehat J_n)$.  We conclude that for every $\bs x\in \R^d$
\begin{eqnarray}
\label{eq:j1} (\lambda^*_{\min}+o_P(1))\|\bs x\|^2&\leq&\|\bs x\|^2_{J_n}\;\leq\;(\lambda^*_{\max}+o_P(1))\|\bs x\|^2\\
\label{eq:j2}  (\lambda^*_{\min}+o_P(1))\|\bs x\|^2&\leq&\|\bs x\|^2_{\widehat J_n}\;\leq\;(\lambda^*_{\max}+o_P(1))\|\bs x\|^2.
\end{eqnarray}
%\piotr{I commented out the lemma that was here. I did not like how we were using it and I preferred direct arguments.}
%  result. 
% \begin{lem}\label{lem:equivJ}
% For any sequence of random vectors $U_n$ in $\R^d$ we have that  $\|U_n\|_{J_n}^2=O_P(\|U_n\|^2)$  and $\|U_n\|^2=O_P(\|U_n\|_{J_n}^2)$. The same holds for the sequence $\widehat J_n$.
% \end{lem}
%\steffen{The first sentence here did not make sense so I have commented it out. } %Recall that $\Theta_v$ is the projection of $\Theta$ onto $\bs\theta_v$ and $M_u$ is the projection of $M$ on $\bs \mu_u$. 
% Recall that \steffen{$C_u$?} $M_u'=M_u\cap A_u$ and $\Theta_v'=\Theta_v\cap A_v$ define restrictions on $\bs\psi=(\bs \mu_u,\bs\theta_v)$ that define the mixed family $\cE'$. 

Now let
$$F_n(\bs \psi)=-\ell(\bs \psi;\bs t_n)\qquad\mbox{and}\qquad \widecheck F_n(\bs \psi)=-\widecheck\ell(\bs \psi;\widehat{\bs \mu}_n).$$ The MLE $\widetilde{\bs \psi}_n$ is the minimizer of $F_n$ over $M_u'\times \Theta_v'$ and $\widecheck{\bs \psi}_n$ is the minimizer of $\widecheck F_n$ over $M_u'\times \Theta_v$ (equivalently over $M_u'\times \Theta_v'$, as argued in Theorem~\ref{th:identity}, because $\widehat{\bs\psi}_n\in M_u\times \Theta_v'$). 

Recall that $U_n=O_P(1)$ denotes that $U_n$ is bounded in probability, that is, for every $\epsilon>0$ there exists $M\in \R$ such that $\P(\|U_n\|>M)<\epsilon$. 
\begin{lem}\label{lem:ratecheck}It holds that $\sqrt{n}(\overline{\bs\psi}_n-\bs\psi_0)=O_P(1)$,  $\sqrt{n}(\widetilde{\bs\psi}_n-\bs\psi_0)=O_P(1)$, and 
 $\sqrt{n}(\widecheck{\bs\psi}_n-\bs\psi_0)=O_P(1)$.
\end{lem}
\begin{proof}
By Corollary~\ref{cor:indistr}, $\sqrt{n}(\overline{\bs\psi}_n-\bs\psi_0)$ and $\sqrt{n}(\widetilde{\bs \psi}_n-\bs\psi_0)$ converge in distribution and thus they are bounded in probability. To show the same for the other estimator, we use the fact that locally it is obtained by suitably projecting $\widehat{\bs \psi}_n$ on $M_u'\times \Theta_v$  and $\sqrt{n}(\widehat{\bs \psi}_n-\bs\psi_0)$ is bounded in probability because it converges in distribution by Corollary~\ref{cor:indistr}. More formally, since $\widehat{\bs \psi}_n$ is the global minimizer of $\widecheck F_n$,  $\widecheck{\bs \psi}_n$ is the minimizer over $M_u'\times \Theta_v$, and $\bs\psi_0\in M_u'\times \Theta_v$, we get 
$$
0\leq \widecheck F_n(\widecheck{\bs \psi}_n)-\widecheck F_n(\widehat{\bs \psi}_n)\leq \widecheck F_n(\bs \psi_0)-\widecheck F_n(\widehat{\bs \psi}_n).
$$
Using the second-order expansion of $\widecheck F_n(\bs \psi)$ around $\bs \psi=\widehat{\bs \psi}_n$ and Lemma~\ref{lem:infor}, we get
$$
0\;\leq\;\tfrac{1}{2}\|\widecheck{\bs \psi}_n-\widehat{\bs \psi}_n\|^2_{\widehat J_n}+o_P(\|\widecheck{\bs \psi}_n-\widehat{\bs \psi}_n\|^2)\;\leq\;\tfrac{1}{2}\|\widehat{\bs \psi}_n-\bs\psi_0\|^2_{\widehat J_n}+o_P(\|\widehat{\bs \psi}_n-\bs\psi_0\|^2).
$$
Here we have also used the basic fact that $R(h)=o(\|h\|^2)$ implies $R(U_n)=o_P(\|U_n\|^2)$; cf.\ Lemma~2.12 in \cite{vandervaart}. 
Multiply this inequality by $n$ so that the right hand side becomes
$$
\tfrac{1}{2}\|\sqrt{n}(\widehat{\bs \psi}_n-\bs\psi_0)\|^2_{\widehat J_n}+o_P(\|\sqrt{n}(\widehat{\bs \psi}_n-\bs\psi_0)\|^2).
$$ Since $\sqrt{n}(\widehat{\bs \psi}_n-\bs \psi_0)$ is bounded in probability, 
 so is $\|\sqrt{n}(\widehat{\bs \psi}_n-\bs\psi_0)\|^2$. 
 Using \eqref{eq:j2} and the fact that $o_P(O_P(1))=o_P(1)$ we conclude that the right hand side is bounded in probability and so the left hand side is bounded too, that is, $\sqrt{n}(\widecheck{\bs\psi}_n-\widehat{\bs \psi}_n)=O_P(1)$. The triangle inequality now implies $\sqrt{n}(\widecheck{\bs\psi}_n-\bs \psi_0)=O_P(1)$, which concludes the proof.
\end{proof}

We are now ready to prove the Theorem. 
\begin{proof}[Proof of Theorem~\ref{th:equal}]
Since, by Corollary~\ref{cor:indistr}, $\sqrt{n}(\widetilde{\bs\psi}_n-\bs\psi_0)$ converges in distribution, it remains to show that $\sqrt{n}(\widetilde{\bs\psi}_n-\widecheck{\bs \psi}_n)=o_P(1)$; see \cite[Theorem 3]{mann1943stochastic}. The standard local first order conditions for optimality of $\widetilde{\bs \psi}_n$ are
\begin{equation}\label{eq:optmuhat}
    \nabla F_n(\widetilde{\bs \psi}_n)^T\cdot (\bs \psi-\widetilde{\bs \psi}_n)\;\geq\;0\qquad \mbox{for all }\bs \psi\in M_u'\times \Theta_v',
\end{equation}
expressing that the directional derivative in any feasible direction must be non-negative.
First-order Taylor  expansion at $\overline{\bs\psi}_n$ gives
$$
\nabla F_n(\widetilde{\bs \psi}_n)=J_n\cdot (\widetilde{\bs \psi}_n-\overline{\bs \psi}_n)+o_P(\|\widetilde{\bs \psi}_n-\overline{\bs \psi}_n\|),
$$
where we used Lemma~\ref{lem:infor} and the fact that $\nabla F_n(\overline{\bs \psi}_n)=0$. After multiplying by $\sqrt{n}$ the last term becomes $o_P(\sqrt{n}\|\widetilde{\bs \psi}_n-\overline{\bs \psi}_n\|)$. By the triangle inequality and Lemma~\ref{lem:ratecheck}, we have $\sqrt{n}\|\widetilde{\bs \psi}_n-\overline{\bs \psi}_n\|\leq \sqrt{n}\|\widetilde{\bs \psi}_n-\bs \psi_0\|+\sqrt{n}\|\bs\psi_0-\overline{\bs \psi}_n\|= O_P(1)$. 
Using the fact that $o_P(O_P(1))=o_P(1)$, we get
\begin{equation}\label{eq:taylorexpress}
\sqrt{n}\nabla F_n(\widetilde{\bs \psi}_n)=J_n\cdot \sqrt{n}(\widetilde{\bs \psi}_n-\overline{\bs \psi}_n)+o_P(1).
\end{equation}
Multiply \eqref{eq:optmuhat} by $n$ and insert the expression in (\ref{eq:taylorexpress}) for  $\sqrt{n}\nabla F_n(\widetilde{\bs \psi}_n)$   to conclude that for all $\bs \psi\in M_u'\times \Theta_v'$
\begin{equation}\label{eq:optmuhat2}
   \sqrt{n}(\widetilde{\bs \psi}_n-\overline{\bs \psi}_n)^T\cdot J_n\cdot \sqrt{n}(\bs \psi-\widetilde{\bs \psi}_n)+\|\sqrt{n}({\bs \psi}-\widetilde{\bs\psi}_n)\|o_P(1)\;\geq\;0.
\end{equation}
Here we used a basic fact that if $\bs r_n$ is a sequence of random vectors in a finite dimensional vector space such that each coordinate is $o_P(1)$ and $\bs q_n$ is another sequence of random vectors then $\bs r_n^T \bs q_n\leq\|\bs q_n\| o_P(1)$.

Similarly, local optimality conditions for $\widecheck{\bs \psi}_n$ are
$$
\nabla \widecheck F_n(\widecheck{\bs \psi}_n)^T\cdot (\bs \psi-\widecheck{\bs \psi}_n)\;\geq \;0\qquad \mbox{for all }\bs \psi\in M_u'\times \Theta_v.
$$
 By Lemma~\ref{lem:infor}, we have $\nabla \widecheck F_n(\widehat{\bs \psi}_n)=0$ and $\nabla^2 \widecheck F_n(\widehat{\bs \psi}_n)=\widehat J_n=J_n+o_P(1)$ (cf. \eqref{eq:Jhat}), which gives that 
$$
\nabla \widecheck F_n(\widecheck{\bs \psi}_n)\;=\;J_n\cdot (\widecheck{\bs \psi}_n-\widehat{\bs \psi}_n)+o_P(\|\widecheck{\bs \psi}_n-\widehat{\bs \psi}_n\|),
$$
and so, again using Lemma~\ref{lem:ratecheck}, the first order optimality conditions for $\widecheck{\bs \psi}_n$, we get that for all $\bs \psi\in M_u'\times \Theta_v$
\begin{equation}\label{eq:optmucheck}
   \sqrt{n}(\widecheck{\bs \psi}_n-\widehat{\bs \psi}_n)^T\cdot J_n\cdot \sqrt{n}(\bs \psi-\widecheck{\bs \psi}_n)+\|\sqrt{n}(\widecheck{\bs \psi}_n-\bs\psi)\|o_P( 1)\;\geq\;0.
\end{equation}
Note also that the optimality conditions for $\widehat{\bs\psi}_n$ are
$$
\nabla F_n(\widehat{\bs\psi}_n)^T\cdot (\psi-\widehat{\bs \psi}_n)\geq 0\qquad\mbox{for all }\bs\psi\in M_u\times \Theta'_v.
$$
Like in the previous two cases, we argue that for all $\bs\psi\in M_u\times \Theta'_v$
\begin{equation}\label{eq:optmuhat3}
   \sqrt{n}(\widehat{\bs \psi}_n-\overline{\bs \psi}_n)^T\cdot J_n\cdot \sqrt{n}(\bs \psi-\widehat{\bs \psi}_n)+\|\sqrt{n}({\bs \psi}-\widehat{\bs\psi}_n)\|o_P(1)\;\geq\;0.
\end{equation}

The rest of the proof relies on using \eqref{eq:optmuhat2}, \eqref{eq:optmucheck}, and \eqref{eq:optmuhat3} evaluated at various points $\bs\psi$ in order to show that $\sqrt{n}(\widetilde{\bs\psi}_n-\widecheck{\bs \psi}_n)=o_P(1)$. We now drop dependence on $n$ from our notation to keep it simple. We will exploit the fact that  $J$ is positive definite and has a block-diagonal form with blocks, which we denote by $J_{uu}$ and $J_{vv}$ but the exact form of these blocks given in Lemma~\ref{lem:infor} is irrelevant here.  

Insert $(\widecheck{\bs \mu}_{u},\widetilde{\bs\theta}_v)$ for $\bs\psi$ in \eqref{eq:optmuhat2} and $(\widetilde{\bs \mu}_{u},\widecheck{\bs\theta}_v)$ for $\bs\psi$ in \eqref{eq:optmucheck} noting that both  are valid points in $M_u'\times \Theta'_v$; then add both expressions  and use that the $\bs u$-block of  $\overline{\bs\psi}_n$ and $\widehat{\bs\psi}_n$ are both equal to $\bs u_n$ (cf.\ Theorem~\ref{th:identity}). From this calculation we 
get
$$
-\|\sqrt{n}(\widecheck{\bs \mu}_u-\widetilde{\bs \mu}_u)\|^2_{J_{uu}}+ \|\sqrt{n}(\widecheck{\bs \mu}_u-\widetilde{\bs\mu}_u)\|\cdot o_P(1)\;\geq\;0.$$
 This can be rewritten as  
$$
0\leq \|\sqrt{n}(\widecheck{\bs \mu}_u-\widetilde{\bs \mu}_u)\|^2_{J_{uu}}\leq  \|\sqrt{n}(\widecheck{\bs \mu}_u-\widetilde{\bs\mu}_u)\|\cdot o_P(1).$$
Using a version of \eqref{eq:j1} for $J_{uu}$, we conclude that $\sqrt{n}(\widetilde{\bs \mu}_u-\widecheck{\bs \mu}_u)=o_P(1)$. 

%\color{asparagus}
Equation \eqref{eq:optmuhat2} evaluated at  $\bs\psi=(\widetilde{\bs \mu}_{u},\widecheck{\bs\theta}_v)\in M_u'\times \Theta_v'$ (note that $\widecheck{\bs\theta}_v=\widehat{\bs\theta}_v$ by Theorem~\ref{th:identity}) yields
$$
  \sqrt{n}(\widetilde{\bs \theta}_v-\overline{\bs \theta}_v)^T\cdot J_{vv}\cdot \sqrt{n}(\widehat{\bs \theta}_v-\widetilde{\bs \theta}_v)+\|\sqrt{n}(\widehat{\bs \theta}_v-\widetilde{\bs\theta}_v)\|o_P(1)\;\geq\;0.
$$
Similarly, insert $\bs\psi=(\bs u,\widetilde{\bs\theta}_v)=(\widehat{\bs\mu}_u,\widetilde{\bs\theta}_v)\in M_u\times \Theta_v'$  into \eqref{eq:optmuhat3} to get
$$
\sqrt{n}(\widehat{\bs\theta}_v-\overline{\bs\theta}_v)^T \cdot J_{vv}\cdot \sqrt{n}(\widetilde{\bs\theta}_v-\widehat{\bs\theta}_v)+\|\sqrt{n}(\widetilde{\bs\theta}-\widehat{\bs\theta}_v)\|o_P(1)\geq 0. 
$$
Adding these two inequalities yields 
$$
0\leq \|\sqrt{n}(\widecheck{\bs \theta}_v-\widetilde{\bs \theta}_v)\|^2_{J_{vv}}\leq  \|\sqrt{n}(\widecheck{\bs \theta}_v-\widetilde{\bs\theta}_v)\|\cdot o_P(1).$$%\;\geq\;0.
%$$
Using a version  of \eqref{eq:j1} for  $J_{vv}$, we conclude that $\sqrt{n}(\widetilde{\bs \theta}_v-\widecheck{\bs \theta}_v)=o_P(1)$. This finally gives that $\sqrt{n}(\widetilde{\bs \psi}_n-\widecheck{\bs \psi}_n)=o_P(1)$, as desired.
\end{proof}

\section{The single-linkage matrix}\label{app:Zmatrix}

We first define the single-linkage matrix of a covariance matrix $S$. Let $R$ be a symmetric $p\times p$ positive semidefinite matrix such that $R_{ii}=1$ for all $i=1,\ldots,p$. In our case $R$ will be the corresponding correlation matrix of $S$.  Consider the graph $G^+$ over $V=\{1,\ldots,d\}$ with an edge between $i$ and $j$ whenever $R_{ij}>0$. Assign to each edge the corresponding positive weight $R_{ij}$ and note that $G^+$ in general does not have to be connected. Define a $d\times d$ matrix $Z$ by setting $Z_{ii}=1$ for all $i$ and
\begin{equation}\label{eq:Z}
	Z_{ij}\;\;:=\;\;\max_{P} \min_{uv\in P} R_{uv}, 
	\end{equation}
where the maximum is taken over all paths $P$ in $G^+$ between $i$ and $j$ and is set to zero if no such path exists. We call $Z$ the \emph{single-linkage matrix} of $R$.
\begin{prop}[\cite{LUZ}]\label{prop:Z}
	Let $R$ be a symmetric $d\times d$ positive semidefinite matrix satisfying $R_{ii}=1$ for all $i=1,\ldots,p$. Then the single-linkage matrix $Z$ of $R$ has ones on the diagonal and satisfies $Z\geq R$. If, in addition,  $R_{ij}<1$ for all $i\neq j$, then $Z$ is an inverse M-matrix.
\end{prop}
Now if $S$ is a symmetric positive semidefinite matrix with strictly positive entries on the diagonal and such that $S_{ij}<\sqrt{S_{ii}S_{jj}}$. Then, by Proposition~\ref{prop:Z}, there exists an inverse M-matrix $Z$ such that $Z\geq S$ and $Z_{ii}=S_{ii}$ for all $i=1,\ldots,d$, obtained by appropriate scaling of the correlation matrix $R$ of $S$. This matrix is called the {single-linkage matrix} of $S$. If $S$ is a sample covariance matrix based on at least two observations, the single-linkage matrix is always positive definite with probability one.  We are now ready to show that the positive graphical lasso estimate exists for $n\geq 2$ observations.

\begin{proof}[Proof of Theorem~\ref{th:plassoexists}]

We must construct a feasible point $\Sigma^0$ of the dual problem \eqref{eq:pospenalRdual} in the case when $L$ is the zero matrix. 
If $S$ is positive definite, we can take $\Sigma^0=S$ so assume that $S$ is rank deficient. The single-linkage matrix $Z$ of $S$ by construction satisfies $Z\geq S$. If the entries of $U$ are sufficiently large then $S+L\leq Z\leq S+U$ and so $Z$ is dually feasible. If $Z$ is not upper bounded by $S+U$ we proceed as follows. Let  $\rho=\min_{i\neq j}U_{ij}$. Since $\|Z-S\|_\infty>0$, we can define 
\begin{equation}\label{eq:Sigma0}
\Sigma^0\;:=\;(1-t^*)S+t^* Z,\quad\mbox{ where } t^*\;=\;\min\left\{1,\frac{\rho}{\|Z-S\|_\infty}\right\}
\end{equation}
which lets $\Sigma^0=Z$ if $\rho=\infty$. Then $\Sigma^0\geq S$ and it is equal to $S$ on the diagonal. Moreover, 
$$\|\Sigma^0-S\|\;=\;t^*\|Z-S\|_\infty\;\leq\;\rho$$
and hence $\Sigma^0$ is dually feasible. Since $\Sigma_0$ is dually feasible, the optimum exists. This concludes the proof.
\end{proof}

Finally, we comment briefly on computational issues. Computing Z can be done efficiently using the link of this construction to single-linkage clustering; cf.\ \cite[Proposition 3.7]{LUZ}. More precisely, we first take the corresponding correlation matrix $R$ and form a dissimilarity matrix $D$, where 
$$
D_{ij}\;=\;\begin{cases}
-\log R_{ij} & \mbox{if }R_{ij}>0\\
\infty & \mbox{otherwise}.
\end{cases}
$$
By construction $D_{ii}=0$ for all $i\in V$. We then run the single linkage clustering algorithm on $D$. The time complexity of this step is $O(d^2)$. The R function \texttt{hclust} by default does not return the underlying ultrametric matrix $\widehat D$ of distances but this information can be recovered from the standard output with a bit of work. Now the single-linkage matrix of $R$ is simply
$$
Z_{ij}\;=\; \exp(-\widehat D_{ij}).
$$
The single linkage matrix of $S$ is obtained by rescaling the matrix $Z$ with the diagonal entries of $S$.
\end{appendix}

\begin{acks}[Acknowledgments]
 We are grateful to Robert Castelo for providing us with an interesting dataset and for patiently explaining the underlying biology to us.
\end{acks}
%%%%%%%%%%%%%%%%%%%%%%%%%%%%%%%%%%%%%%%%%%%%%%
%% Funding information, if any,             %%
%% should be provided in the                %%
%% funding section.                         %%
%%%%%%%%%%%%%%%%%%%%%%%%%%%%%%%%%%%%%%%%%%%%%%
\begin{funding}
 The second author was supported in part by  the Spanish Government grants (RYC-2017-22544,PGC2018-101643-B-I00), and Ayudas Fundaci\'on BBVA a Equipos de Investigaci\'on Cientifica 2017.
\end{funding}

\bibliographystyle{imsart-number}
\bibliography{bib_mtp2}
\end{document}